\documentclass[pdflatex,sn-basic]{sn-jnl}

\usepackage{graphicx}
\usepackage{booktabs}
\usepackage{multirow}
\usepackage{amsmath,amssymb,amsfonts}
\usepackage{amsthm}
\usepackage{mathtools}
\usepackage{algorithm}
\usepackage{algpseudocode}
\usepackage{placeins}
\usepackage{microtype}
\usepackage{url}
\hypersetup{hypertexnames=false}

\theoremstyle{thmstyleone}
\newtheorem{theorem}{Theorem}
\newtheorem{proposition}{Proposition}
\newtheorem{corollary}{Corollary}
\newtheorem{lemma}{Lemma}
\theoremstyle{thmstylethree}
\newtheoremstyle{compactassumption}%
  {6pt plus1pt minus1pt}
  {6pt plus1pt minus1pt}
  {\small\normalfont}
  {0pt}
  {\small\bfseries}
  {.}
  {.5em}
  {}
\theoremstyle{compactassumption}
\newtheorem{assumption}{Assumption}
\theoremstyle{thmstyletwo}
\newtheorem{remark}{Remark}

\newcommand{\bbeta}{\boldsymbol{\beta}}
\newcommand{\btheta}{\boldsymbol{\theta}}
\newcommand{\bx}{\boldsymbol{x}}
\newcommand{\bg}{\boldsymbol{g}}
\newcommand{\bz}{\boldsymbol{z}}

\newcommand{\bA}{\boldsymbol{A}}
\newcommand{\bB}{\boldsymbol{B}}
\newcommand{\bC}{\boldsymbol{C}}
\newcommand{\bD}{\boldsymbol{D}}
\newcommand{\bI}{\boldsymbol{I}}
\newcommand{\bTheta}{\boldsymbol{\Theta}}
\newcommand{\R}{\mathbb{R}}
\newcommand{\Hcal}{\mathcal{H}}
\newcommand{\Bcal}{\mathcal{B}}
\newcommand{\Scal}{\mathcal{S}}
\DeclareMathOperator{\supp}{supp}
\DeclareMathOperator{\trim}{Trim}

\begin{document}

\title[Heterogeneity-calibrated Byzantine-robust CQR]{Heterogeneity-calibrated Byzantine-robust distributed composite quantile regression}

\author[1]{Xiaofei Wu}\email{xfwu1016@ynu.edu.cn}
\author*[2]{Jian Qing Shi}\email{jianqingshi@bnbu.edu.cn}
\affil[1]{\orgdiv{Yunnan Key Laboratory of Statistical Modeling and Data Analysis, School of Mathematics and Statistics},
  \orgname{Yunnan University},
  \orgaddress{\city{Kunming}, \country{China}}}
\affil*[2]{\orgdiv{Department of Statistics and Data Science},
  \orgname{Beijing Normal--Hong Kong Baptist University},
  \orgaddress{\city{Zhuhai}, \country{China}}}

\abstract{We study sparse composite quantile regression (CQR) for distributed data with heterogeneous honest sites and Byzantine workers. Honest sites share a common slope but may differ in their covariate distributions, error laws, and quantile intercepts. The proposed heterogeneity-calibrated robust CQR (HC-RCQR) profiles local intercepts and calibrates scores using an approximate inverse profile Hessian. Honest workers transmit the resulting vectors, whereas Byzantine workers may send arbitrary vectors. The server updates the estimate by coordinatewise trimming and soft thresholding. A scalar example shows how unequal honest-site curvatures allow intermediate Byzantine reports to survive trimming and how ideal calibration reduces their possible effect. We also establish nonidentification of the mean of unrestricted honest-site slopes when fault identities are unknown. Under suitable conditions, we establish conditional contraction and support-recovery guarantees. The bound separates score offset, sampling fluctuation, contamination, calibration error, and the Newton remainder. Simulations and a bike-demand study examine performance under heterogeneous data and adversarial messages.}



\maketitle

\section{Introduction}\label{sec:intro}

Quantile regression characterizes conditional distributions without requiring finite response variance \citep{koenker1978regression}. By pooling several quantile levels, composite quantile regression (CQR) can improve efficiency while retaining robustness to heavy-tailed errors \citep{zou2008composite}; sparsity regularization extends it to problems with many predictors \citep{belloni2011quantile,gu2020sparsecqr}. Distributed quantile regression and distributed sparse CQR have been studied for data stored across multiple sites, using local estimates or score information rather than pooled records \citep{chen2020distributedqr,chen2023distributedcqr,pan2022distributed}.

A distributed system, however, may be neither homogeneous nor fully trustworthy. Honest sites can differ in their covariate distributions, error laws, scales, and quantile intercepts \citep{duan2022heterogeneity,jin2025nonrandom}, whereas Byzantine workers may transmit arbitrary vectors \citep{yin2018optimal,allouah2023mixing}.
After profiling out the site-specific intercepts, local profile curvature measures how the slope score changes with the slope. Away from the target, differences in covariate and error distributions can therefore produce different honest scores for the same slope error. A Byzantine report between these scores may survive coordinatewise trimming. In the scalar example of Proposition~\ref{prop:rawattack}, the largest displacement from the honest-score average has a first-order term proportional to the curvature difference times the slope error.

Existing quantile methods generally address only one side of this problem. Communication-efficient quantile regression, CQR, and pilot-based refinements assume compliant workers \citep{chen2020distributedqr,chen2023distributedcqr,pan2022distributed}. Byzantine procedures based on coordinatewise medians or trimmed means have instead been analyzed mainly under homogeneous sampling \citep{yin2018optimal,tu2021vrmom}, with model-specific extensions to support vector machines and sparse surrogate CQR \citep{wang2025svm,chen2024byzantine}. In particular, the closely related estimator of \citet{chen2024byzantine} assumes identically distributed honest samples. Conversely, non-random-partition CQR allows honest-site heterogeneity but assumes protocol-compliant workers \citep{jin2025nonrandom,jin2026poisson}; it therefore does not address the Byzantine setting considered here.

Related Byzantine-learning methods handle heterogeneous gradients through bucketing, mixing, or splitting \citep{karimireddy2022bucketing,allouah2023mixing,liu2023splitting}. Heterogeneity-aware statistical inference uses site-specific nuisance parameters or transported scores but assumes protocol-compliant sites \citep{duan2022heterogeneity}. Byzantine-resilient multi-task feature selection \citep{wang2026mtfl} and federated policy evaluation \citep{wang2026brfedtd} address other learning problems. Our focus is the combination of local quantile-intercept estimation, curvature calibration, and robust aggregation in sparse CQR.

We propose heterogeneity-calibrated robust CQR (HC-RCQR) for a common sparse
slope with site-specific designs, error laws, and intercepts.
Proposition~\ref{prop:identification} explains why an average of unrestricted
site-specific slopes is not identifiable when fault identities are unknown.
HC-RCQR profiles local quantile intercepts and calibrates slope scores using the
Schur-complement Hessian before server trimming and soft thresholding. Each
honest site sends one $p$-vector per refinement round; its intercepts and
curvature matrix remain local (Fig.~\ref{fig:byzantine}).

\begin{figure}[t]
\centering
\includegraphics[width=0.96\linewidth]{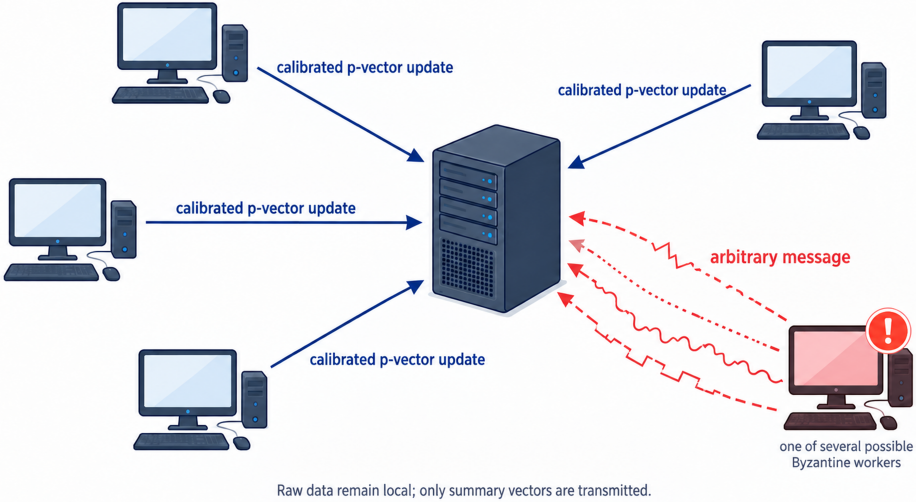}
\caption{Distributed learning with heterogeneous honest workers and one illustrative Byzantine worker. Honest sites transmit calibrated $p$-vector updates while keeping raw data, quantile intercepts, and curvature matrices local; Byzantine sites may transmit arbitrary messages. The model allows more than one Byzantine worker.}
\label{fig:byzantine}
\end{figure}

Conditional on uniform score, calibration, and initialization events, Theorem~\ref{thm:recursion} and Corollary~\ref{cor:rate} give a nonasymptotic contraction bound and support recovery when nonzero coefficients are sufficiently large. The bound separates calibrated-score offset, honest-average fluctuation, contamination, calibration error, and the Newton remainder. Simulations and a bike-demand study compare complete procedures under sign-reversal and heterogeneity-mimicking attacks. The organization of this paper is as follows. Section~\ref{sec:model} introduces the model and identification argument; Section~\ref{sec:method} develops the procedure; Section~\ref{sec:theory} gives the conditional theory; Sections~\ref{sec:sim} and~\ref{sec:bike} report numerical results; and Section~\ref{sec:discussion} discusses limitations and extensions; the appendices collect proofs, simulation and application details, a finite-dimensional verification, and matched calibration comparisons.

\section{Model, target, and identification}\label{sec:model}

\subsection{Heterogeneous common-slope quantile model}

There are $m$ workers. A fixed but unknown honest set $\Hcal$ has size $m_H$, and $\Bcal=\{1,\ldots,m\}\setminus\Hcal$ denotes the Byzantine set. Honest worker $k$ stores independent observations $\{(Y_{ki},\bx_{ki}):1\leq i\leq n_k\}$, with $\bx_{ki}\in\R^p$. Across honest sites the observations are independent but need not be identically distributed. For fixed quantile levels $0<\tau_1<\cdots<\tau_q<1$, assume
\begin{equation}
 Q_{\tau_\ell}(Y_{ki}\mid\bx_{ki})=b^*_{k\ell}+\bx_{ki}^{\mathsf T}\bbeta^*,
 \qquad k\in\Hcal,\quad \ell=1,\ldots,q.             \label{eq:model}
\end{equation}
The same sparse slope $\bbeta^*\in\R^p$ is imposed across the honest sites and the selected quantile levels, with support $\Scal=\supp(\bbeta^*)$ and $s=|\Scal|$. The intercepts $b^*_{k\ell}$ are site- and quantile-specific. The common-slope restriction therefore specifies a single structural target at the selected quantiles; Assumption~\ref{ass:curvature} provides the curvature needed for the local analysis. A simple compatible model is $Y_{ki}=\bx_{ki}^{\mathsf T}\bbeta^*+\varepsilon_{ki}$, where the error law may vary across sites but is independent of the covariates. Its site-specific error quantiles are absorbed by the intercepts. Site-dependent or quantile-dependent slopes are outside this paper. Thus, their covariate distributions and conditional error laws may have different scales, tail behavior, and densities at the selected quantiles, and their intercepts may differ as well, provided these differences do not alter the slope in \eqref{eq:model}.

We allow Byzantine workers to send arbitrary messages, including coordinated and history-dependent ones. We require only $|\Bcal|\leq \alpha m$ for some known upper bound $\alpha<1/2$. They need not generate data from model~\eqref{eq:model}. Raw records remain at their originating sites, and only summary vectors are exchanged, following the standard distributed-inference setup \citep{chen2020distributedqr,chen2023distributedcqr,duan2022heterogeneity}.

\subsection{Why an honest-site slope average is not identifiable}

One might instead allow a different slope $\bbeta_k$ at every honest site and combine these slopes, as in the meta-analytic framework of \citet{becker2007synthesis}. Consider their arithmetic mean over honest sites. When the server does not know which sites are honest, it cannot generally recover this mean even from exact population-slope reports. Proposition~\ref{prop:identification} gives two configurations with identical reports and different honest-site means.

\begin{proposition}[Identification boundary]\label{prop:identification}
Fix integers $m\geq3$ and $1\leq b<m/2$. Suppose exactly $b$ workers are Byzantine; their number is known but their identities are not. Each honest worker reports its population slope without sampling error, whereas each Byzantine worker may report any vector. If the honest-site slopes may vary arbitrarily in $\R^p$, then the target
\[
 \boldsymbol\theta_{\Hcal}=\frac{1}{m-b}\sum_{k\in\Hcal}\bbeta_k
\]
is not identifiable from the worker-labelled reports. More precisely, for every nonzero $\boldsymbol a\in\R^p$, two admissible configurations produce identical reports but targets $-\boldsymbol a/(m-b)$ and $\boldsymbol a/(m-b)$. Every deterministic estimator therefore has Euclidean error at least $\|\boldsymbol a\|_2/(m-b)$ in one configuration.
\end{proposition}

\noindent\textit{Proof.}
Write $m_H=m-b$. Let a set $S$ of $m_H-1$ workers report zero, let two further workers $u$ and $v$ report $-\boldsymbol a$ and $\boldsymbol a$, and let the remaining $b-1$ workers also report zero. In configuration A, the honest set is $S\cup\{u\}$; in configuration B, it is $S\cup\{v\}$. Both configurations have exactly $b$ Byzantine workers and the same ordered reports, but their honest means differ by $2\boldsymbol a/m_H$. The triangle inequality gives the stated error bound. $\square$

For $m=3$ and $b=1$, the reports are $(-\boldsymbol a,\boldsymbol0,\boldsymbol a)$ and the honest means are $\pm\boldsymbol a/2$. Under an upper bound $\alpha$, the construction requires $\lfloor\alpha m\rfloor\geq1$. Other targets and structured slope models may be identifiable, but need their own identification argument. Here we retain a common slope and allow heterogeneous designs, error laws, and intercepts.

\section{Heterogeneity-calibrated robust CQR}\label{sec:method}

\subsection{Smoothed and profiled local loss}

For $\rho_\tau(u)=u\{\tau-\mathbf{1}(u<0)\}$, the empirical CQR loss is non-differentiable. To obtain stable local Newton steps, we use the logistic convolution
\begin{equation}
 \rho_{\tau,h}(u)=\tau u+h\log\{1+\exp(-u/h)\},          \label{eq:smoothloss}
\end{equation}
where $h>0$ is a bandwidth. This convex, twice differentiable approximation follows the convolution-smoothed construction of \citet{he2023smoothed}. Its derivative is $\psi_{\tau,h}(u)=\tau-\{1+\exp(u/h)\}^{-1}$.

At honest site $k$, the intercept vector $\boldsymbol b=(b_1,\ldots,b_q)^{\mathsf T}$
is estimated locally for each candidate slope. The smoothed loss and its profile version are
\begin{equation}
 L_{k,h}(\bbeta,\boldsymbol{b})=
 \frac{1}{n_kq}\sum_{i=1}^{n_k}\sum_{\ell=1}^q
 \rho_{\tau_\ell,h}(Y_{ki}-b_\ell-\bx_{ki}^{\mathsf T}\bbeta),
 \qquad
 \widehat Q_{k,h}(\bbeta)=\min_{\boldsymbol{b}\in\R^q}L_{k,h}(\bbeta,\boldsymbol{b}).      \label{eq:profile}
\end{equation}
For a fixed candidate $\bbeta$, the minimization in (\ref{eq:profile}) separates over the
$q$ quantile levels. Each local intercept solves
$\sum_{i=1}^{n_k}\psi_{\tau_\ell,h}
(Y_{ki}-\widehat b_{k\ell}(\bbeta)-\bx_{ki}^{\mathsf T}\bbeta)=0$ and is then
substituted into $L_{k,h}$. Hence
$\widehat Q_{k,h}(\bbeta)=L_{k,h}\bigl(\bbeta,
\widehat{\boldsymbol b}_k(\bbeta)\bigr)$ is a slope-only objective. The intercepts
are re-estimated whenever the slope changes and remain local. Adding a constant
to every response at a site shifts its fitted intercepts by the same constant
and leaves its profiled loss unchanged.

Let $\widehat{\bg}_k(\bbeta)=\nabla\widehat Q_{k,h}(\bbeta)$ and
$\widehat{\bA}_k(\bbeta)=\nabla^2\widehat Q_{k,h}(\bbeta)$. Since the intercept
score is zero at the profile solution, the envelope theorem gives
$\widehat{\bg}_k(\bbeta)=\nabla_{\bbeta}L_{k,h}
\bigl(\bbeta,\widehat{\boldsymbol b}_k(\bbeta)\bigr)$. For curvature, partition
the full Hessian at this point using the ordering
$(\boldsymbol b,\bbeta)$:
\[
 \begin{pmatrix}\widehat{\bD}_k&\widehat{\bC}_k\\
 \widehat{\bC}_k^{\mathsf T}&\widehat{\bB}_k\end{pmatrix},
\]
where the blocks are the intercept--intercept, intercept--slope, and slope--slope
blocks. Differentiating the local first-order conditions gives
$\partial\widehat{\boldsymbol b}_k(\bbeta)/\partial\bbeta^{\mathsf T}
=-\widehat{\bD}_k^{-1}\widehat{\bC}_k$, which yields the profiled Hessian
as the Schur complement
\begin{equation}
 \widehat{\bA}_k=\widehat{\bB}_k-
 \widehat{\bC}_k^{\mathsf T}\widehat{\bD}_k^{-1}\widehat{\bC}_k.       \label{eq:schur}
\end{equation}
Thus, \eqref{eq:schur} accounts for how the fitted intercepts change with the slope. The Hessian block $\widehat{\bB}_k$ alone holds the intercepts fixed and generally differs from this profile curvature.

\subsection{Local curvature calibration}

The following scalar example illustrates how curvature heterogeneity can let an intermediate Byzantine report survive trimming. With three workers and $1/3\leq\gamma<1/2$, the rule removes one report from each tail.
\begin{proposition}[Attack displacement under unequal curvature]\label{prop:rawattack}
Consider two honest scalar profile losses $Q_1,Q_2$ and one Byzantine worker reporting any $v\in\R$. Suppose $Q_k'(\beta^*)=0$, $a_k=Q_k''(\beta^*)>0$, and $Q_k''$ is $L_k$-Lipschitz near $\beta^*$. At $\beta=\beta^*+e$ in this neighborhood, set $g_k=Q_k'(\beta)$ and $\bar g_H=(g_1+g_2)/2$. If $v$ lies strictly between $g_1$ and $g_2$, trimming one value from each tail removes both honest reports and retains $v$. The worst displacement from the honest-score average satisfies
\[
 \sup_{v\in\R}|\operatorname{median}(g_1,g_2,v)-\bar g_H|
 =\frac{|g_1-g_2|}{2}
 \geq\frac{|a_1-a_2|}{2}|e|-\frac{L_1+L_2}{4}e^2.
\]
For ideal messages $z_k=\beta-a_k^{-1}g_k$, instead,
\[
 \sup_{v\in\R}|\operatorname{median}(z_1,z_2,v)-\beta^*|
 =\max_{k=1,2}|z_k-\beta^*|
 \leq\max_{k=1,2}\frac{L_k}{2a_k}e^2.
\]
\end{proposition}
For $a_1\ne a_2$ and sufficiently small nonzero $e$, unequal curvature creates
a first-order difference between honest scores. The first bound concerns
worst-case displacement from their average; the second bounds parameter error
relative to $\beta^*$. Ideal calibration removes the first-order term.
A retained Byzantine report may be harmless or helpful in one realization, but
it cannot be trusted across rounds: the worker may adapt its message, and the
server cannot distinguish a helpful report from one chosen to bias a later
update.
Appendix~\ref{sup:aggregation} provides the proof and a smoothed CQR example satisfying the
stationarity condition $Q_k'(\beta^*)=0$.

At round $t$, after receiving the server iterate $\bbeta^{(t)}$, honest worker $k$ forms
\begin{equation}
\bz_k^{(t)}=\bbeta^{(t)}-\widehat{\bTheta}_k^{(t)}
 \widehat{\bg}_k(\bbeta^{(t)}),                         \label{eq:message}
\end{equation}
where $\widehat{\bTheta}_k^{(t)}$ controls the residual $\bI_p-\widehat{\bTheta}_k^{(t)}\widehat{\bA}_k(\bbeta^{(t)})$ on sparse directions. Multiplying the local profile score by this approximate inverse is a Newton-type inverse-curvature calibration: it puts site-specific scores on a common parameter scale before aggregation. In low or moderate dimensions, we use
$\widehat{\bTheta}_k^{(t)}=\{\widehat{\bA}_k(\bbeta^{(t)})+\kappa\bI_p\}^{-1}$ with a small ridge $\kappa$. For $\kappa>0$, convexity makes the profile Hessian positive semidefinite, so the ridge matrix is positive definite even when the Hessian is singular. Profiling one intercept per quantile gives $\operatorname{rank}\{\widehat{\bA}_k(\bbeta^{(t)})\}\leq\min\{p,q(n_k-1)\}$. Thus $p\geq n_k$ alone does not imply singularity; the bound implies it when $p>q(n_k-1)$, while local collinearity may cause it earlier. For large $p$, the row-wise Dantzig construction below avoids a dense inverse and targets sparse directions. An $\ell_1$-constrained approximate inverse can instead be constructed row by row using the Dantzig-type program \citep{candes2007dantzig}
\begin{equation}
 \widehat{\btheta}_{kj}^{(t)}=\arg\min_{\btheta\in\R^p}\|\btheta\|_1
 \quad\text{subject to}\quad
 \|\widehat{\bA}_k(\bbeta^{(t)})\btheta-\boldsymbol{e}_j\|_\infty\leq\nu_k,          \label{eq:dantzig}
\end{equation}
where $\boldsymbol{e}_j\in\R^p$ is the $j$th canonical basis vector. Following
\citet{cai2011clime}, set row $j$ of $\widehat{\bTheta}_k^{(t)}$ to
$(\widehat{\btheta}_{kj}^{(t)})^{\mathsf T}$. Symmetry of the Hessian then gives
\[
 \|\bI_p-\widehat{\bTheta}_k^{(t)}\widehat{\bA}_k(\bbeta^{(t)})\|_{\max}\leq\nu_k.
\]
Here $\|M\|_{\max}=\max_{j,l}|M_{jl}|$ and
$\|M\|_{\infty\to\infty}=\max_j\sum_l|M_{jl}|$. Subsequent symmetrization
need not preserve these constraints and is not used. Appendix~\ref{sup:notation} proves
the row-wise bound.

The program in \eqref{eq:dantzig} supplies the approximate inverse for
\eqref{eq:message}. To see its effect, let
$\boldsymbol e^{(t)}=\bbeta^{(t)}-\bbeta^*$. Taylor expansion gives
\[
 \bz_k^{(t)}-\bbeta^*
 =\{\bI_p-\widehat{\bTheta}_k^{(t)}\widehat{\bA}_k(\bbeta^{(t)})\}\boldsymbol e^{(t)}
 -\widehat{\bTheta}_k^{(t)}\widehat{\bg}_k(\bbeta^*)
 +\boldsymbol R_{k,t},
\]
where Assumption~\ref{ass:calibration} gives
$\|\boldsymbol R_{k,t}\|_\infty\leq M_A L_A\|\boldsymbol e^{(t)}\|_\infty^2$
on its event and neighborhood. The first term is the residual from approximate
inversion. If the empirical profile Hessian is invertible and its exact inverse
satisfies the same row-norm bound, this term vanishes and the Newton message satisfies
\begin{equation}
 \bz_k^{(t)}-\bbeta^*=-\widehat{\bA}_k(\bbeta^{(t)})^{-1}\widehat{\bg}_k(\bbeta^*)
 +O(\|\boldsymbol e^{(t)}\|_\infty^2).                         \label{eq:cancellation}
\end{equation}
Equation~\eqref{eq:cancellation} describes exact-Newton cancellation. With an
approximate inverse, Assumption~\ref{ass:calibration} bounds the additional term
by $\mu_n\|\boldsymbol e^{(t)}\|_\infty$, giving the $\mu_n r$ contribution in
Lemma~\ref{lem:message}. Calibration does not equalize message variances;
heterogeneity still affects score distributions and remainder bounds.

\subsection{Robust aggregation and sparsity}

Once the local messages have been calibrated, the server aggregates them as follows. For $m$ scalars $v_1,\ldots,v_m$, let $\trim_\gamma(v_1,\ldots,v_m)$ be the mean after deleting the $\lfloor\gamma m\rfloor$ largest and smallest values. Since $|\Bcal|\leq\lfloor\alpha m\rfloor\leq\lfloor\gamma m\rfloor$ and $m-2\lfloor\gamma m\rfloor>0$, the trimming count covers the allowed faults and leaves a nonempty set. Applied coordinatewise to the messages, this gives
\[
 \widetilde{\bz}^{(t)}=\trim_\gamma\{\bz_1^{(t)},\ldots,\bz_m^{(t)}\},
 \qquad \alpha\leq\gamma<1/2.
\]
The server update is
\begin{equation}
 \bbeta^{(t+1)}=\mathcal S_{\lambda_t}(\widetilde{\bz}^{(t)}),
 \qquad
 \{\mathcal S_\lambda(v)\}_j=\operatorname{sign}(v_j)(|v_j|-\lambda)_+. \label{eq:update}
\end{equation}
Here $(u)_+=\max(u,0)$. To initialize, each honest worker fits an
$\ell_1$-penalized profiled CQR criterion and transmits its coefficient vector;
Byzantine workers may send arbitrary vectors. The server takes the
coordinatewise median and soft-thresholds it at $\lambda_{\rm init}$ to obtain
$\bbeta^{(0)}$. The local penalty is used for this pilot fit. Refinement instead
uses the unpenalized profile score in \eqref{eq:message}, with sparsity imposed
by the server after aggregation.

\begin{algorithm}[t]
\caption{Heterogeneity-calibrated robust CQR (HC-RCQR)}\label{alg:hcrcqr}
\begin{algorithmic}[1]
\Require Quantiles $\{\tau_\ell\}_{\ell=1}^q$, bandwidth $h$, rounds $T\geq1$, trim fraction $\gamma$, local penalties, initial threshold $\lambda_{\rm init}$, thresholds $\{\lambda_t\}_{t=0}^{T-1}$, inverse-Hessian tuning $\kappa$ or $\nu_k$
\State Honest workers compute penalized local profile-CQR estimates; Byzantine workers may send arbitrary vectors. The server takes the coordinatewise median of all reports and soft-thresholds at $\lambda_{\rm init}$ to obtain $\bbeta^{(0)}$.
\For{$t=0,\ldots,T-1$}
  \State The server broadcasts $\bbeta^{(t)}$.
  \State Honest worker $k$ profiles $\widehat{\boldsymbol b}_k(\bbeta^{(t)})$, computes $\widehat{\bg}_k$ and $\widehat{\bA}_k$, forms $\widehat{\bTheta}_k$, and returns $\bz_k^{(t)}$ in \eqref{eq:message}.
  \State Byzantine workers return arbitrary $p$-vectors.
  \State The server computes the coordinatewise trimmed mean of the $m$ messages and applies \eqref{eq:update}.
\EndFor
\State \Return $\widehat{\bbeta}=\bbeta^{(T)}$.
\end{algorithmic}
\end{algorithm}

Honest sites send one $p$-vector at initialization and per round, for
$O\{mp(T+1)\}$ transmitted numbers; intercepts and Hessians stay local.
The theory requires the initializer to belong to the sparse neighborhood
$\mathcal C(r_0)$ in Assumption~\ref{ass:curvature}. Appendix~\ref{sup:iteration} gives a
deterministic median-and-thresholding guarantee under a bound on honest pilot
errors; Section~\ref{sec:theory} states the corresponding probability requirement.

\subsection{Practical implementation}\label{sec:implementation}

The experiments use quantiles $\{0.25,0.50,0.75\}$, ridge inversion, and fixed tuning constants. For the row-wise construction in \eqref{eq:dantzig}, numerical feasibility verifies the imposed residual constraint; the row-norm and Hessian-approximation conditions in Assumption~\ref{ass:calibration} also enter the theory. Before the run, the designer supplies a contamination upper bound $\alpha<1/2$ and chooses $\gamma\in[\alpha,1/2)$. The server does not estimate this bound from messages. The simulations use $\gamma=\min\{\alpha+0.05,0.40\}$.

For the reported simulations, use $h=0.24$, $\kappa=0.08$, four refinement
rounds, local penalty $0.24\sqrt{\log(p)/n_k}$, threshold floor
$\lambda_{\rm floor}=0.035+0.035\alpha$, initialization threshold
$\lambda_{\rm init}=1.5\lambda_{\rm floor}$, and
$\lambda_t=(1+0.5^{t+1})\lambda_{\rm floor}$. The numerical implementation caps the Euclidean norm of each Newton step at $2.5\max\{1,\|\bbeta^{(t)}\|_2\}$; the theory analyzes uncapped steps. The application uses its own
sample-size and bandwidth values, which are listed in Appendix~\ref{sup:bike}.

\section{Conditional nonasymptotic analysis}\label{sec:theory}

Detailed proofs of results not proved below are provided in Appendices~\ref{sup:notation}--\ref{sup:identification}.
We next formalize the calibration argument. Fix a failure probability
$\delta\in(0,1)$. In the rate notation, $x\lesssim y$ means that
$x\leq Cy$ for a constant $C$ independent of the relevant sample-size and
dimensional parameters, while $x\asymp y$ means that both
$x\lesssim y$ and $y\lesssim x$ hold. The usual $O(\cdot)$ notation has the
corresponding upper-bound meaning. Since site sizes may differ, let
$N_H=\sum_{k\in\Hcal}n_k$, $n_{\min}=\min_{k\in\Hcal}n_k$, and
$n_{\max}=\max_{k\in\Hcal}n_k$. To keep the subscripts concise, the local-size
index $n$ in $a_{n,h}$ and $\varepsilon_{m,n,p}$ denotes the scale $n_{\min}$
in the heterogeneous case. Let $a_{n,h}$ denote the calibrated profile-score
offset bound specified below, and define
\begin{equation}
 \varepsilon_{m,n,p}(\delta)=a_{n,h}+
 \sqrt{\frac{\log(2p/\delta)}{N_H}}+
 \frac{\gamma}{1-2\gamma}
 \sqrt{\frac{\log(2pm/\delta)}{n_{\min}}}.              \label{eq:epsilon}
\end{equation}
For notational brevity, write $\varepsilon_{m,n,p}=\varepsilon_{m,n,p}(\delta)$ below. For an asymptotic sequence with $mp\to\infty$, one concrete convention is $\delta=(mp)^{-c_\delta}$ for a fixed $c_\delta>0$; both logarithmic factors in \eqref{eq:epsilon} are then of order $\log(mp)$.
The second term bounds fluctuation of the honest-worker average at the pooled sample-size scale; the third accounts for trimming and contamination. When
$\gamma\asymp\alpha$ and $\gamma\leq\bar\gamma<1/2$, the third term has the
scale
\[
\alpha\sqrt{\frac{\log(2pm/\delta)}{n_{\min}}},
\]
or, up to logarithmic factors, $\alpha/\sqrt{n_{\min}}$
\citep{yin2018optimal}. Under the second-order offset condition in
Assumption~\ref{ass:calibration}, $a_{n,h}\lesssim h^2$.

Assumptions~\ref{ass:sampling}--\ref{ass:curvature} collect familiar sampling, quantile-density, sparsity, and local-curvature requirements. The common slope and the bound $\alpha<1/2$ are structural identification and robust-aggregation conditions, respectively. Site-size comparability is used only to express an unweighted honest-worker average at the pooled scale; Appendix~\ref{sup:notation} gives the corresponding expression without it. The inverse row-sum bound in Assumption~\ref{ass:curvature} is stronger than eigenvalue control and is needed for coordinatewise error bounds. Assumption~\ref{ass:calibration} is instead the principal high-probability uniform calibration condition: it requires the approximate inverse, empirical Hessian, and profiled score bounds to hold simultaneously over all honest sites and candidate iterates in $\mathcal C(r_0)$.

\begin{assumption}[Sampling and common slope]\label{ass:sampling}
Within each honest site the observations are independent and identically distributed, and honest sites are mutually independent. They satisfy \eqref{eq:model}, $\mathbb E|Y_{ki}|<\infty$, and
$\sup_{\|\boldsymbol u\|_2=1}\|\boldsymbol u^{\mathsf T}\bx_{ki}\|_{\psi_2}\leq K_x$
uniformly over honest $k$, where $\|Z\|_{\psi_2}=\inf\{c>0:\mathbb E\exp(Z^2/c^2)\leq2\}$. Site sizes are comparable, $n_{\max}/n_{\min}\leq C_n$, and the common slope has $s$ nonzero entries. Byzantine messages may be arbitrary and history dependent, with $|\Bcal|\leq\alpha m$ and $\alpha\leq\gamma\leq\bar\gamma<1/2$.
\end{assumption}

\begin{assumption}[Quantile regularity]\label{ass:density}
Let $R_{ki}=Y_{ki}-\bx_{ki}^{\mathsf T}\bbeta^*$, and let $f_k(u\mid\bx)$ denote the conditional density of $R_{ki}$ at site $k$. For constants $c_f,\underline f,\overline f,L_f>0$, uniformly over honest $k$, selected $\ell$, almost every covariate value, and $|v|\leq c_f$,
\[
 \underline f\leq f_k(b^*_{k\ell}+v\mid\bx)\leq\overline f,
 \qquad |\partial_v f_k(b^*_{k\ell}+v\mid\bx)|\leq L_f.
\]
The quantile levels remain in a compact subset of $(0,1)$.
\end{assumption}

\begin{assumption}[Profile curvature]\label{ass:curvature}
Let
$Q_{k,h}(\bbeta)=\min_{\boldsymbol b\in\R^q}\mathbb E[q^{-1}\sum_{\ell=1}^q
\rho_{\tau_\ell,h}(Y_k-b_\ell-\bx_k^{\mathsf T}\bbeta)]$
be the population profiled loss, and write $\bA_k(\bbeta)=\nabla^2Q_{k,h}(\bbeta)$, with $h$ suppressed from the notation. On the sparse neighborhood
\[
 \mathcal C(r_0)=\{\bbeta:\|\bbeta-\bbeta^*\|_\infty\leq r_0,\ 
 |\supp(\bbeta)\cup\Scal|\leq c_s s\},
\]
where $r_0>0$ and $c_s\geq1$, all $\bA_k(\bbeta)$ have eigenvalues in $[c_A,C_A]$, these matrix-valued maps are uniformly Lipschitz in the operator norm induced by $\ell_\infty$, and
$\sup_{k\in\Hcal,\,\bbeta\in\mathcal C(r_0)}\|\bA_k(\bbeta)^{-1}\|_{\infty\to\infty}\leq M_A$. The constants are common across sites, although the matrices themselves may differ.
\end{assumption}

\begin{assumption}[Uniform local calibration and score control]\label{ass:calibration}
For a matrix $\boldsymbol M$, define its restricted action norm by
$\|\boldsymbol M\|_{\infty,r}=\sup\{\|\boldsymbol M\boldsymbol v\|_\infty:
\|\boldsymbol v\|_\infty\leq1,\ |\supp(\boldsymbol v)|\leq r\}$.
For the specified $\delta$, suppose that there is an event
$\mathcal E_\delta$ with $\mathbb P(\mathcal E_\delta)\geq1-\delta$ on which
the following three requirements hold uniformly over honest sites and
$\bbeta,\bbeta'\in\mathcal C(r_0)$. \emph{First}, for a deterministic $\mu_n\geq0$,
the local inverse-calibration requirement is
\begin{equation}
 \big\|\bI_p-\widehat{\bTheta}_k(\bbeta)
 \widehat{\bA}_k(\bbeta)\big\|_{\infty,c_s s}\leq\mu_n,
 \qquad \|\widehat{\bTheta}_k(\bbeta)\|_{\infty\to\infty}\leq 2M_A. \label{eq:calibrationevent}
\end{equation}
\emph{Second}, for a constant $L_A>0$, the empirical profile Hessian is uniformly Lipschitz:
\[
 \|\widehat{\bA}_k(\bbeta)-\widehat{\bA}_k(\bbeta')\|_{\infty\to\infty}
 \leq L_A\|\bbeta-\bbeta'\|_\infty.
\]
Every segment joining $\bbeta^*$ to a point in $\mathcal C(r_0)$ remains
in $\mathcal C(r_0)$ by its definition, so the same bound applies along that
segment. \emph{Third}, uniformly over the same neighborhood, write
$-\widehat{\bTheta}_k(\bbeta)\widehat{\bg}_k(\bbeta^*)
=\boldsymbol d_{k,h}(\bbeta)+\boldsymbol\xi_k(\bbeta)$.
There are deterministic $a_{n,h}\geq0$ and $C_\xi>0$ such that, uniformly on the same neighborhood,
\begin{align*}
 \max_{k\in\Hcal}\|\boldsymbol d_{k,h}(\bbeta)\|_\infty&\leq a_{n,h},\\
 \left\|m_H^{-1}\sum_{k\in\Hcal}\boldsymbol\xi_k(\bbeta)\right\|_\infty
 &\leq C_\xi\sqrt{\frac{\log(2p/\delta)}{N_H}},\\
 \max_{k\in\Hcal}\|\boldsymbol\xi_k(\bbeta)\|_\infty
 &\leq C_\xi\sqrt{\frac{\log(2pm/\delta)}{n_{\min}}}.
\end{align*}
Here $\boldsymbol d_{k,h}$ is a bounded offset; the decomposition does not
require it to equal an expectation or $\boldsymbol\xi_k$ to be mean zero.
For the second-order specialization used below, we additionally assume
$a_{n,h}\leq C_bh^2$ for a constant $C_b>0$.
\end{assumption}

For ridge inversion, with all quantities evaluated at the same $\bbeta$, the
identity $\bI-\widehat{\bTheta}_k\widehat{\bA}_k
=\kappa\widehat{\bTheta}_k$ and the second bound in
\eqref{eq:calibrationevent} allow $\mu_n$ to be taken as $2\kappa M_A$.
For the row-wise construction in \eqref{eq:dantzig}, Appendix~\ref{sup:notation} gives
a deterministic sufficient condition under which feasibility yields the
inverse row-sum bound and allows $\mu_n$ to be taken as
$c_s s\max_k\nu_k$.

Assumption~\ref{ass:calibration} is an explicit condition for the results below; ridge computability or Dantzig feasibility alone does not verify this uniform event.

\begin{lemma}[Calibrated robust message bound]\label{lem:message}
Under Assumptions~\ref{ass:sampling}--\ref{ass:calibration}, there is a constant $C>0$ such that, on $\mathcal E_\delta$, simultaneously for every $\bbeta\in\mathcal C(r_0)$ with $r=\|\bbeta-\bbeta^*\|_\infty$,
\begin{equation}
 \left\|\trim_\gamma\{\bz_1(\bbeta),\ldots,\bz_m(\bbeta)\}-\bbeta^*\right\|_\infty
 \leq C\left\{\varepsilon_{m,n,p}(\delta)+\mu_n r+r^2\right\}.       \label{eq:messagebound}
\end{equation}
The bound is uniform over all Byzantine strategies satisfying the fraction constraint.
\end{lemma}

The proof of Lemma~\ref{lem:message} is provided in Appendix~\ref{sup:expansion}. The terms $r^2$ and $\mu_n r$ arise from Taylor expansion and approximate inversion. The offset $a_{n,h}$ includes smoothing and nuisance-estimation effects; heterogeneity also enters the calibration conditions and constants.

\begin{theorem}[Iterative estimation and sparsity]\label{thm:recursion}
Suppose Assumptions~\ref{ass:sampling}--\ref{ass:calibration} hold, and define
\begin{equation}
 B(r)=C\{\varepsilon_{m,n,p}(\delta)+\mu_n r+r^2\},                  \label{eq:Br}
\end{equation}
where $C$ is at least the constant in Lemma~\ref{lem:message}. Fix a finite horizon $T$ and deterministic radii $\bar r_t\geq0$. On $\mathcal E_\delta$, suppose that, for every $t=0,\ldots,T-1$, $\bbeta^{(t)}\in\mathcal C(r_0)$ and $\|\bbeta^{(t)}-\bbeta^*\|_\infty\leq\bar r_t$. If $\lambda_t\geq B(\bar r_t)$ for all such $t$, then the update \eqref{eq:update} obeys simultaneously
\begin{align}
 \supp(\bbeta^{(t+1)})&\subseteq\Scal,                                      \label{eq:noFP}\\
 \|\bbeta^{(t+1)}-\bbeta^*\|_\infty&\leq 2\lambda_t,                         \label{eq:linf}\\
 \|\bbeta^{(t+1)}-\bbeta^*\|_2&\leq 2\sqrt{s}\lambda_t,\qquad
 \|\bbeta^{(t+1)}-\bbeta^*\|_1\leq 2s\lambda_t.                              \label{eq:l2l1}
\end{align}
Thus, with probability at least $1-\delta$, these conclusions hold whenever the
stated iterate and threshold premises are satisfied. No additional union bound
over rounds is needed because $\mathcal E_\delta$ is uniform over $\mathcal C(r_0)$.
\end{theorem}

\noindent\textit{Proof.}
Let $\widetilde{\bz}^{(t)}$ be the trimmed message and set
$a_t=\|\widetilde{\bz}^{(t)}-\bbeta^*\|_\infty$. Lemma~\ref{lem:message} gives
$a_t\leq B(\|\bbeta^{(t)}-\bbeta^*\|_\infty)\leq B(\bar r_t)\leq\lambda_t$. If $j\notin\Scal$, then
$|\widetilde z_j^{(t)}|\leq\lambda_t$, and soft thresholding sets that coordinate to zero. For $j\in\Scal$, the scalar inequality
$|\mathcal S_\lambda(u)-v|\leq|u-v|+\lambda$ gives an error no larger than
$a_t+\lambda_t\leq2\lambda_t$. Finally, \eqref{eq:l2l1} follows from
\eqref{eq:noFP}. $\square$

Theorem~\ref{thm:recursion} is therefore a uniform one-step implication, not an unconditional assertion for arbitrary iterations. Corollary~\ref{cor:rate} verifies its neighborhood and radius premises inductively under the stated smallness conditions.

\begin{corollary}[Conditional terminal rate and support recovery]\label{cor:rate}
Under Assumptions~\ref{ass:sampling}--\ref{ass:calibration}, suppose $C\mu_n\leq1/8$,
 $Cr_0\leq1/8$, $4C\varepsilon_{m,n,p}\leq r_0$, and the initializer belongs to $\mathcal C(r_0)$. Define the deterministic radii $\bar r_0=r_0$ and
$\bar r_{t+1}=2B(\bar r_t)$, and take $\lambda_t=B(\bar r_t)$. On $\mathcal E_\delta$, all iterates remain in $\mathcal C(r_0)$ and satisfy $\|\bbeta^{(t)}-\bbeta^*\|_\infty\leq\bar r_t$. After
$T=\max\{1,\lceil\log_2\{r_0/(4C\varepsilon_{m,n,p})\}\rceil\}$ rounds,
the sup-norm error is at most $8C\varepsilon_{m,n,p}$. In particular, for
uniformly bounded $C$,
\begin{equation}
 \|\widehat{\bbeta}-\bbeta^*\|_\infty=O(\varepsilon_{m,n,p}),\quad
 \|\widehat{\bbeta}-\bbeta^*\|_2=O(\sqrt{s}\,\varepsilon_{m,n,p}),\quad
 \|\widehat{\bbeta}-\bbeta^*\|_1=O(s\,\varepsilon_{m,n,p}).          \label{eq:rates}
\end{equation}
If $\min_{j\in\Scal}|\beta_j^*|>2B(\bar r_{T-1})$, then
$\supp(\widehat{\bbeta})=\Scal$.
\end{corollary}

\noindent\textit{Proof.}
Whenever $\bar r_t\leq r_0$, the stated constants give
$\bar r_{t+1}\leq2C\varepsilon_{m,n,p}+\bar r_t/2\leq r_0$.
Induction and Theorem~\ref{thm:recursion} therefore keep the iterates in the sparse neighborhood and yield
$\bar r_t\leq4C\varepsilon_{m,n,p}+2^{-t}r_0$. This proves \eqref{eq:rates} after the stated number of rounds. Moreover, \eqref{eq:noFP} excludes false positives, while the beta-min condition and \eqref{eq:linf} prevent any true coordinate from being zero at round $T$. $\square$

The results above use the uncapped update and the analytical threshold condition $\lambda_t\geq B(\bar r_t)$. The experiments use preset thresholds and cap large Newton steps for stability; Appendix~\ref{sup:repro} gives the resulting extra term. If $\mu_n=0$, the linear calibration remainder disappears.

Appendix~\ref{sup:iteration} proves that Algorithm~\ref{alg:hcrcqr}'s initializer belongs to $\mathcal C(r_0)$ whenever all honest local fits are within $a_0$ in sup norm, $\lambda_{\rm init}\geq a_0$, and $2\lambda_{\rm init}\leq r_0$. If that local event has probability at least $1-\delta_0$, the terminal conclusions hold with probability at least $1-\delta-\delta_0$. A fixed-dimensional scalar CQR example satisfying the calibration and initialization requirements is given in Appendix~\ref{sup:verification}.

When the sites are balanced, $n_k=n$, $\gamma\asymp\alpha$, $a_{n,h}\lesssim h^2$, the uniform envelopes remain bounded, and $\delta$ is fixed, the rate in \eqref{eq:rates} has the leading scale
\begin{equation}
 h^2+\sqrt{\frac{\log(2p)}{m_Hn}}+\alpha\sqrt{\frac{\log(2pm)}{n}},       \label{eq:balanced}
\end{equation}
under the smallness conditions of Corollary~\ref{cor:rate}. If $\delta=(mp)^{-c_\delta}$, both logarithmic factors are instead of order $\log(mp)$. In the uncontaminated case, $\alpha=0$, and without trimming, $h^2\lesssim\sqrt{\log(2p)/N_H}$ gives the pooled sparse rate on the stated events. By contrast, if $\alpha>0$ and the local sample size $n$ is fixed, the displayed contamination term does not vanish as the number of workers grows, consistent with the worker-contamination mechanism studied by \citet{yin2018optimal}. The displayed result is an upper bound and does not establish minimax optimality for heterogeneous CQR. It also does not imply asymptotic normality at the pooled root-$N_H$ rate when the Byzantine fraction is fixed. Valid inference requires additional conditions, for example a sufficiently fast vanishing contamination fraction, and is left for future work.

\section{Simulation study}\label{sec:sim}

\subsection{Design and competitors}

We designed the simulation to separate the effects of heterogeneity from those of malicious communication. Specifically, we used $m=20$ workers, $n=90$ observations per worker, $p=30$ predictors, and the five-sparse slope
\[
 \bbeta^*=(1.00,-0.82,0.66,-0.46,0.24,0,\ldots,0)^{\mathsf T}.
\]
At worker $k$, the covariate scales, an equicorrelated component, and a pre-standardization location shift varied smoothly with the phase $2\pi(k-1)/m$. Student-$t_3$ noise was standardized and then assigned a site-specific location and scale. A heterogeneity index $H\in\{0,0.5,0.8,1\}$ multiplied all between-site departures, while the common slope remained fixed as required by \eqref{eq:model}. Covariates were then standardized within each site; correlation and response-error differences remained. Appendix~\ref{sup:simulation} gives the generating equations.

We then considered two Byzantine strategies drawn from attack families used in Byzantine-learning evaluations \citep{allen-zhu2021byzantine,allouah2023mixing}. Under sign reversal, each malicious message equaled minus four times the coordinatewise median honest message, with a small perturbation when several workers were malicious. This is a scaled sign-flipping attack; the factor four is an experimental stress-test choice \citep{allen-zhu2021byzantine}. The heterogeneity-mimic attack was subtler and followed the near-center construction of the A Little Is Enough (ALIE) attack and the heterogeneity-aware mimic attack \citep{baruch2019alie,karimireddy2022bucketing}: each malicious value was placed $0.9$ empirical standard deviations from the honest coordinatewise median in the direction suggested by the honest mean. Thus these messages were designed to lie near the honest coordinatewise center, although the formula does not guarantee that every malicious coordinate remains within the realized honest range. Both attacks are constructed after observing the current honest messages and therefore represent omniscient stress tests.

HC-RCQR used the tuning specified in Section~\ref{sec:implementation}. We
compared it with four procedures: Oracle, an infeasible pooled estimator using
honest records and site-specific intercepts; Mean, which averages attacked local
CQR estimates; Med-L, which takes their coordinatewise median and hard-thresholds;
and RTM-score, which iterates a trimmed mean of raw scores without calibration.
RTM-score is included as a raw-score comparator; it is not a full reimplementation of \citet{chen2024byzantine}.
The procedures use different initializations and iteration budgets, so this is a
comparison of complete methods. Estimation errors use the fitted coefficients;
support recovery selects those with absolute value exceeding $0.08$. Results
average 50 independently seeded replications.

\subsection{Results}

Table~\ref{tab:sim} and Fig.~\ref{fig:sim} show the effects of heterogeneity and attacks. Under 20\% sign reversal, HC-RCQR's mean $\ell_2$ error increased from $0.152$ at $H=0$ to $0.191$ at $H=1$. At $H=1$, this error was about 52\% lower than Med-L and 72\% lower than RTM-score. Because the methods use different initializations and tuning parameters, this comparison does not isolate the effect of calibration.

The heterogeneity-mimic attack places messages near the honest coordinatewise median and was less damaging to the simple mean in this design. This setting evaluates estimation accuracy when malicious messages are close to honest reports. At $H=0.8$ and $\alpha=0.2$, HC-RCQR had mean $\ell_2$ error $0.116$, compared with $0.323$ for Mean, $0.321$ for Med-L, and $0.527$ for RTM-score. Its error was still larger than that of the infeasible Oracle; this gap may reflect the combined finite-sample costs of local summaries, regularization, and adversarial protection. The corresponding exact-support frequencies are shown in Fig.~\ref{fig:support}; HC-RCQR was at 1.00 (100\%) across the displayed settings, whereas exact-support recovery for Med-L and RTM-score declined as heterogeneity increased, and that of RTM-score also declined as contamination increased.

\begin{table}[t]
\caption{Simulation performance over 50 replications. Entries for $\ell_2$ and $\ell_1$ error are mean (Monte Carlo standard error); Exact support is the recovery frequency; Selected is the mean model size, both using the cutoff $0.08$. Oracle uses the honest records and is not implementable. Lower error and larger exact-support frequency are better}\label{tab:sim}
\centering
\small
\setlength{\tabcolsep}{2.8pt}
\begin{tabular}{llrrrr}
\toprule
Setting & Method & $\ell_2$ error & $\ell_1$ error & Exact support & Selected \\
\midrule
$H=0,\ \alpha=.2$ & Oracle & 0.059 (0.002) & 0.131 (0.005) & 1.00 & 5.00 \\
 & Mean & 1.545 (0.001) & 3.221 (0.003) & 0.00 & 0.00 \\
 & Med-L & 0.266 (0.003) & 0.590 (0.007) & 0.96 & 4.96 \\
 & RTM-score & 0.458 (0.005) & 1.011 (0.011) & 0.52 & 4.52 \\
 & HC-RCQR & 0.152 (0.002) & 0.333 (0.005) & 1.00 & 5.00 \\
\addlinespace[2pt]
$H=.5,\ \alpha=.2$ & Oracle & 0.071 (0.002) & 0.159 (0.006) & 1.00 & 5.00 \\
 & Mean & 1.546 (0.001) & 3.247 (0.004) & 0.00 & 0.00 \\
 & Med-L & 0.315 (0.004) & 0.697 (0.008) & 0.84 & 4.84 \\
 & RTM-score & 0.548 (0.005) & 1.206 (0.010) & 0.20 & 4.20 \\
 & HC-RCQR & 0.163 (0.003) & 0.351 (0.007) & 1.00 & 5.00 \\
\addlinespace[2pt]
$H=1,\ \alpha=.2$ & Oracle & 0.090 (0.003) & 0.203 (0.007) & 1.00 & 5.00 \\
 & Mean & 1.569 (0.002) & 3.320 (0.005) & 0.00 & 0.00 \\
 & Med-L & 0.396 (0.006) & 0.875 (0.014) & 0.34 & 4.34 \\
 & RTM-score & 0.689 (0.006) & 1.506 (0.014) & 0.10 & 4.08 \\
 & HC-RCQR & 0.191 (0.004) & 0.410 (0.009) & 1.00 & 5.00 \\
\addlinespace[2pt]
$H=.8,\ \alpha=.2$ (mimic) & Oracle & 0.078 (0.002) & 0.173 (0.005) & 1.00 & 5.00 \\
 & Mean & 0.323 (0.004) & 0.847 (0.009) & 0.96 & 4.96 \\
 & Med-L & 0.321 (0.005) & 0.703 (0.012) & 0.86 & 4.86 \\
 & RTM-score & 0.527 (0.005) & 1.149 (0.011) & 0.56 & 4.56 \\
 & HC-RCQR & 0.116 (0.004) & 0.264 (0.009) & 1.00 & 5.00 \\
\bottomrule
\end{tabular}
\end{table}

\begin{figure}[t]
\centering
\includegraphics[width=0.97\textwidth]{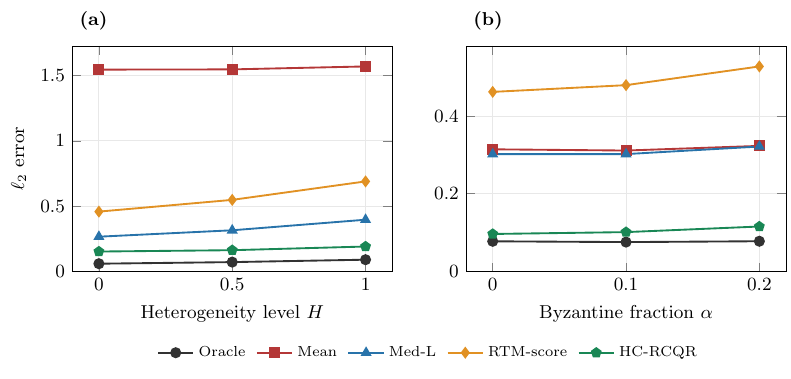}
\caption{Mean $\ell_2$ error over 50 replications. (a) Increasing honest-site heterogeneity under $\alpha=0.2$ sign reversal. (b) Increasing Byzantine fraction at $H=0.8$ under the heterogeneity-mimic attack. Oracle uses honest pooled records; RTM-score robustly aggregates raw scores without curvature calibration. Lines connect only the evaluated settings and are not fitted trends}\label{fig:sim}
\end{figure}
\begin{figure}[t]
\centering
\includegraphics[width=0.97\textwidth]{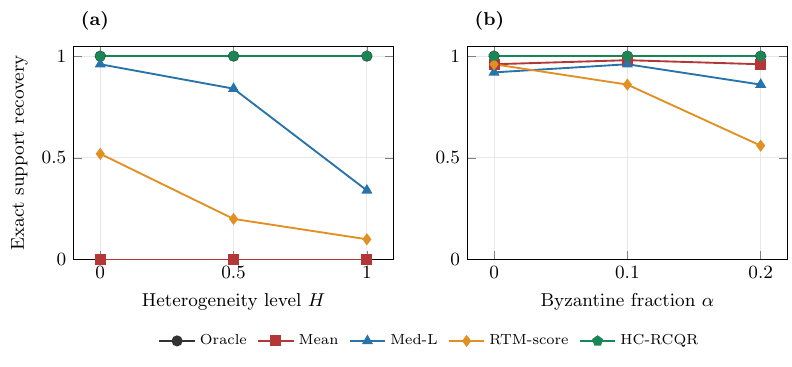}
\caption{Exact-support recovery frequency over 50 replications, selecting coefficients with absolute value exceeding $0.08$. (a) Increasing honest-site heterogeneity under $\alpha=0.2$ sign reversal. (b) Increasing Byzantine fraction at $H=0.8$ under the heterogeneity-mimic attack. Oracle uses honest records and is infeasible.}
\label{fig:support}
\end{figure}
Appendix~\ref{sup:matched} adds 50 paired replications per setting with identical initialization, thresholds, four-round budgets and caps. Only the preconditioner changes, comparing local ridge inversion with $I,4I,8I$. Local calibration has lower mean error than $I$ in all four settings, but under sign reversal at $H=1$, $4I$ attains $0.176$ versus $0.183$ for local calibration. The comparison supports an advantage over the unit step, but not uniformly over all scalar steps.

\section{Bike-demand application}\label{sec:bike}

The hourly file used from the UCI Bike Sharing data set \citep{fanaee2013bike}
contains 17,379 records without missing values. We took $\log(1+\mathtt{cnt})$
as the response and treated the 24 year--month blocks as non-identical sites.
Within each site, the earliest 80\% of observations formed the training sample
and the latest 20\% formed the test sample. Training records from later months
can occur after test records from earlier months, so this is a within-site
held-out evaluation rather than a global forecasting experiment.

The 38 predictors comprised four continuous variables (temperature, apparent
temperature, humidity, and wind speed), holiday and working-day indicators, and
one-hot indicators for hour, weekday, and weather category. To avoid target
leakage, we excluded the casual and registered counts, whose sum is the response,
and raw date identifiers. The continuous variables were standardized using pooled
training records only; this preprocessing was trusted, and only model messages
were attacked. Each site's training-residual quantiles supplied the three
prediction intercepts, so the comparison concerns transfer of the common slope
across months rather than site-specific response levels.

We evaluated one no-attack setting and two attack settings in which four of the 24 workers were Byzantine ($\alpha=1/6$). For each attacked setting, ten seeds independently selected the malicious sites. The method settings and numerical constants were fixed before test evaluation. We compared the five procedures on honest test sites using average composite pinball loss, median-prediction mean absolute error (MAE), and coverage of the interval between the fitted 0.25 and 0.75 quantiles. As in the simulation, Oracle uses records known to be honest and serves only as a benchmark.

As Table~\ref{tab:bike} shows, HC-RCQR was the distributed procedure closest to Oracle in composite loss and median MAE in all three settings. Under sign reversal, its composite loss was $0.244$, compared with $0.345$ for Med-L and $0.400$ for RTM-score; the corresponding median MAE values were $0.558$, $0.791$, and $0.905$. Under the mimic attack, HC-RCQR had a composite loss of $0.240$. At the same time, coverage of the central 50\% interval remained close to the nominal $0.50$ level for all methods, while HC-RCQR produced appreciably narrower intervals: the mean width was $0.857$ under sign reversal and $0.834$ under the mimic attack, compared with $1.257$ and $1.237$ for Med-L. In the no-attack setting, the reported standard errors round to zero because all 24 sites are honest and the deterministic implementation is unaffected by their permutation. Figure~\ref{fig:bikeperf} complements Table~\ref{tab:bike} by displaying the loss--interval-width comparison across the three application settings.

\begin{table}[t]
\caption{Honest-site test performance on the Bike Sharing data set. Entries for loss and MAE are mean (standard error) over ten Byzantine-site assignments; coverage is the empirical coverage of the fitted interquartile interval. Oracle is infeasible. Lower loss and MAE are better}\label{tab:bike}
\centering
\small
\setlength{\tabcolsep}{5pt}
\begin{tabular}{llrrr}
\toprule
Attack & Method & Composite loss & Median MAE & 50\% coverage \\
\midrule
None & Oracle & 0.222 (0.000) & 0.508 (0.000) & 0.489 \\
 & Mean & 0.340 (0.000) & 0.781 (0.000) & 0.503 \\
 & Med-L & 0.340 (0.000) & 0.779 (0.000) & 0.506 \\
 & RTM-score & 0.396 (0.000) & 0.899 (0.000) & 0.509 \\
 & HC-RCQR & 0.236 (0.000) & 0.540 (0.000) & 0.499 \\
\addlinespace[2pt]
Sign reversal & Oracle & 0.222 (0.001) & 0.506 (0.003) & 0.489 \\
 & Mean & 0.421 (0.001) & 0.956 (0.003) & 0.505 \\
 & Med-L & 0.345 (0.001) & 0.791 (0.003) & 0.502 \\
 & RTM-score & 0.400 (0.001) & 0.905 (0.003) & 0.504 \\
 & HC-RCQR & 0.244 (0.001) & 0.558 (0.003) & 0.502 \\
\addlinespace[2pt]
Heterogeneity-mimic & Oracle & 0.224 (0.001) & 0.512 (0.002) & 0.488 \\
 & Mean & 0.341 (0.001) & 0.782 (0.002) & 0.505 \\
 & Med-L & 0.341 (0.001) & 0.783 (0.002) & 0.507 \\
 & RTM-score & 0.398 (0.001) & 0.903 (0.002) & 0.510 \\
 & HC-RCQR & 0.240 (0.001) & 0.551 (0.002) & 0.503 \\
\bottomrule
\end{tabular}
\end{table}

\begin{figure}[t]
\centering
\includegraphics[width=0.97\textwidth]{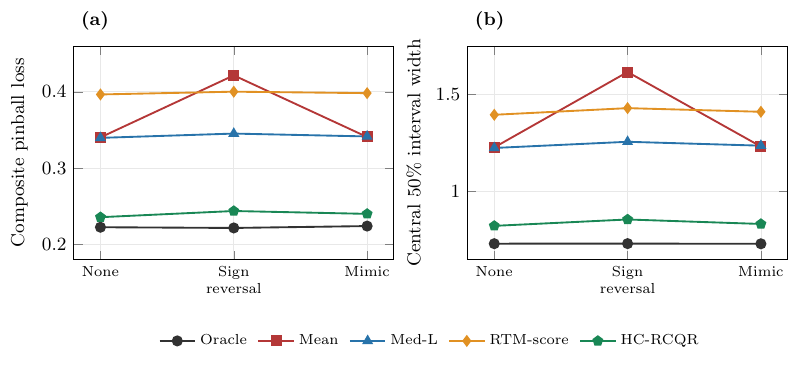}
\caption{Mean application performance across the three Bike Sharing settings. (a) Composite pinball loss. (b) Width of the central 50\% prediction interval. Values are means over the assignment seeds; Oracle is infeasible.}
\label{fig:bikeperf}
\end{figure}

The Byzantine messages are artificially injected; the data contain no identified real attacks. This experiment assesses sensitivity to cross-month heterogeneity under a controlled attack. Some local design columns are constant or redundant, and ridge
inversion makes the updates computable; the experiment does not verify the
sampling, curvature, or common-slope assumptions. Hourly dependence also falls outside the independent-sampling theory; the common slope is a working model for cross-month prediction. The aggregation rules do not use fault labels, but
the experimental penalty uses honest-site sample sizes, which are known from
the attack assignments. Deployment would require tuning that does not use
unknown fault identities and a justified upper bound on the Byzantine fraction.

\FloatBarrier
\section{Conclusion and Discussion}\label{sec:discussion}

This paper studies sparse composite quantile regression with heterogeneous honest
sites and Byzantine workers. We address these two issues jointly by retaining a
common sparse slope while allowing site-specific designs, error laws, and
quantile intercepts. Proposition~\ref{prop:identification} shows why an
average of unrestricted site-specific slopes is not identifiable when fault
identities are unknown. HC-RCQR profiles local intercepts, calibrates slope
scores with local profile Hessians, and uses coordinatewise trimming and soft
thresholding at the server. Proposition~\ref{prop:rawattack} shows how
curvature heterogeneity can allow an intermediate Byzantine report to survive
raw-score trimming. Under the event in Assumption~\ref{ass:calibration} and a
valid initializer, Lemma~\ref{lem:message}, Theorem~\ref{thm:recursion}, and
Corollary~\ref{cor:rate} provide conditional error, contraction, and support
guarantees. The numerical studies illustrate the method under heterogeneous
sites and adversarial messages.

The main theorem is conditional for the uncapped update and analytical threshold schedule. The scalar probability result applies only to its stated fixed-dimensional model, not to the simulation design or to growing dimension. Coordinatewise trimming is tied to the chosen coordinate system, and the bounds do not establish pooled root-$N_H$ inference under fixed positive contamination. The algorithmic formulation applies to smoothed ordinary quantile regression when $q=1$; unsmoothed loss and varying slopes require new analyses, and SCAD and MCP \citep{fan2001scad,zhang2010mcp} require new thresholding arguments. Future work should address growing dimension, unequal site sizes, restricted-information attacks, and tuning rules that do not use fault identities.

\section*{Statements and Declarations}

\subsection*{Funding}
The work of Jian Qing Shi was supported by the National Key R\&D Program of China (Grant No. 2023YFA1011400) and the National Natural Science Foundation of China (Grant No. 12271239).

\subsection*{Competing interests}
The authors declare no competing interests.

\subsection*{Author contributions}
Xiaofei Wu: methodology, software, formal analysis, simulation, and writing---original draft. Jian Qing Shi: conceptualization, methodology, supervision, and writing---review and editing. 

\subsection*{Data availability}
The Bike Sharing data are publicly available from the UCI Machine Learning Repository under DOI: \url{https://doi.org/10.24432/C5W894} and are not redistributed with this article. Replication code and replicate-level results are provided in the reproducibility archive accompanying the submission.

\subsection*{Code availability}
The Python reference implementation, experiment drivers, fixed seeds, generated tables, and replicate-level results are provided in the reproducibility archive accompanying the submission. The code uses NumPy and pandas and includes commands for reproducing every reported result.

\appendix
\numberwithin{theorem}{section}
\numberwithin{lemma}{section}
\numberwithin{proposition}{section}
\numberwithin{corollary}{section}
\numberwithin{assumption}{section}
\numberwithin{remark}{section}

\section{Notation and high-probability events}\label{sup:notation}

The notation follows the main text. The ingredients needed for the proofs are restated here for convenience. There are $m$ workers, a fixed but unknown honest set $\Hcal$, and a Byzantine set $\Bcal$, with
$|\Bcal|\leq\alpha m$. Honest worker $k$ has $n_k$ observations and satisfies
\[
 Q_{\tau_\ell}(Y_{ki}\mid\bx_{ki})
 =b^*_{k\ell}+\bx_{ki}^{\mathsf T}\bbeta^*,
 \qquad \ell=1,\ldots,q.
\]
The same sparse slope $\bbeta^*$ is imposed across honest sites and the selected quantile levels, whereas the intercepts $b^*_{k\ell}$ may vary by site and quantile.
As in Assumption~\ref{ass:sampling} of the main text, we require $\mathbb E|Y_{ki}|<\infty$.
Together with the design moment bound, this makes the population smoothed
loss finite at every finite parameter value. No finite error variance is
required. The target support is $\Scal=\supp(\bbeta^*)$ and $|\Scal|=s$. We write
$m_H=|\Hcal|$, $N_H=\sum_{k\in\Hcal}n_k$,
$n_{\min}=\min_{k\in\Hcal}n_k$, and $n_{\max}=\max_{k\in\Hcal}n_k$. We assume that site sizes are comparable:
$n_{\max}/n_{\min}\leq C_n$ for a fixed constant $C_n$. This condition lets us bound the unweighted honest-worker average at the pooled scale $N_H^{-1/2}$. Without comparable site sizes, the corresponding term is
$m_H^{-1}(\sum_{k\in\Hcal}n_k^{-1})^{1/2}$.
The sparse neighborhood is
\[
 \mathcal C(r_0)=\{\bbeta:\|\bbeta-\bbeta^*\|_\infty\leq r_0,
 |\supp(\bbeta)\cup\Scal|\leq c_ss\}.
\]
We use $x\lesssim y$ to mean $x\leq Cy$ for a constant $C$ independent of
the relevant sample-size and dimensional parameters, and
$x\asymp y$ when both $x\lesssim y$ and $y\lesssim x$ hold.

For bandwidth $h$, define the smoothed check loss and its derivative by
\[
 \rho_{\tau,h}(u)=\tau u+h\log\{1+\exp(-u/h)\},\qquad
 \psi_{\tau,h}(u)=\tau-\{1+\exp(u/h)\}^{-1}.
\]
Let $\widehat Q_{k,h}(\bbeta)$ be the empirical loss after profiling the
$q$ site-specific intercepts. Its score and Hessian are denoted
$\widehat{\bg}_k(\bbeta)$ and $\widehat{\bA}_k(\bbeta)$.
At an iterate $\bbeta$, the honest message is
\[
 \bz_k(\bbeta)=\bbeta-\widehat{\bTheta}_k(\bbeta)
 \widehat{\bg}_k(\bbeta).
\]

For a matrix $\boldsymbol M$, write
$\|\boldsymbol M\|_{\infty,r}=\sup\{\|\boldsymbol M\boldsymbol v\|_\infty:
\|\boldsymbol v\|_\infty\leq1,\ |\supp(\boldsymbol v)|\leq r\}$.
Fix $\delta\in(0,1)$, and let $\mathcal E_\delta$ denote the joint calibration and score event described next. The proof is easiest to follow after separating its two components. First, on $\mathcal E_\delta$, for every honest site and all $\bbeta,\bbeta'$ in the sparse neighborhood $\mathcal C(r_0)$,
\begin{align}
 \|\bI-\widehat{\bTheta}_k(\bbeta)
       \widehat{\bA}_k(\bbeta)\|_{\infty,c_s s}&\leq\mu_n,     \label{sup:eq:cal}\\
 \|\widehat{\bTheta}_k(\bbeta)\|_{\infty\to\infty}&\leq2M_A,  \label{sup:eq:thetabound}\\
 \|\widehat{\bA}_k(\bbeta)-\widehat{\bA}_k(\bbeta')\|_{\infty\to\infty}
 &\leq L_A\|\bbeta-\bbeta'\|_\infty.                          \label{sup:eq:lipschitz}
\end{align}
Second, on the same event and uniformly for $\bbeta\in\mathcal C(r_0)$, decompose the calibrated
score at the target as
\begin{equation}
 -\widehat{\bTheta}_k(\bbeta)\widehat{\bg}_k(\bbeta^*)
 =\boldsymbol d_{k,h}(\bbeta)+\boldsymbol\xi_k(\bbeta),        \label{sup:eq:scoredec}
\end{equation}
where $\|\boldsymbol d_{k,h}(\bbeta)\|_\infty\leq a_{n,h}$ for a deterministic offset bound $a_{n,h}\geq0$, and the residual worker terms satisfy, uniformly over the neighborhood,
\begin{align}
 \left\|m_H^{-1}\sum_{k\in\Hcal}\boldsymbol\xi_k(\bbeta)\right\|_\infty
 &\leq C_\xi\sqrt{\frac{\log(2p/\delta)}{N_H}},                \label{sup:eq:pooled}\\
 \max_{k\in\Hcal}\|\boldsymbol\xi_k(\bbeta)\|_\infty
 &\leq C_\xi\sqrt{\frac{\log(2pm/\delta)}{n_{\min}}}.          \label{sup:eq:max}
\end{align}
By definition, $\mathcal E_\delta$ requires these bounds over the whole
neighborhood rather than only at a fixed point. Consequently, once the event
occurs, the bounds apply to any data-dependent iterate in $\mathcal C(r_0)$.

We assume $\mathbb P(\mathcal E_\delta)\geq1-\delta$. For an asymptotic sequence with $mp\to\infty$, one may, for example, take $\delta=(mp)^{-c_\delta}$ with fixed $c_\delta>0$, provided that the same event is verified at this confidence level. For the second-order smoothing specialization in the main text, we additionally require $a_{n,h}\leq C_bh^2$. The deterministic proof below does not require the residual terms to be independent or mean zero once these inequalities hold. Independence at a fixed nonrandom argument does not, by itself, establish uniform control at adaptive iterates. Nor do bounded quantile scores alone establish the offset order after empirical intercept profiling and estimated inverse calibration.

Sample splitting may help separate inverse estimation from score evaluation, but it still requires bounds for nuisance estimation, smoothing, and uniformity. No such primitive-condition theorem is asserted here. In moderate dimension, ridge inversion gives the deterministic bound
\[
 \|\bI-(\widehat{\bA}_k+\kappa\bI)^{-1}\widehat{\bA}_k\|_{\infty,c_ss}
 \leq\kappa\|(\widehat{\bA}_k+\kappa\bI)^{-1}\|_{\infty\to\infty}.
\]
Thus, under \eqref{sup:eq:thetabound}, the ridge construction permits the choice $\mu_n=2\kappa M_A$; it does not, however, prove \eqref{sup:eq:scoredec}--\eqref{sup:eq:max}. The following deterministic result makes the corresponding Dantzig requirement explicit. The row-wise constraint is the Dantzig-type formulation of \citet{candes2007dantzig}, used here together with the constrained inverse-estimation idea of \citet{cai2011clime}.

\begin{proposition}[Dantzig calibration check]\label{sup:prop:dantzig}
Here $\boldsymbol e_j$ denotes the $j$th canonical basis vector in $\R^p$.
Let $\|\boldsymbol M\|_{\max}=\max_{j\ell}|M_{j\ell}|$. Suppose, uniformly over honest $k$ and $\bbeta\in\mathcal C(r_0)$, that $\widehat{\bA}_k(\bbeta)$ and $\bA_k(\bbeta)$ are symmetric,
\[
 \|\widehat{\bA}_k(\bbeta)-\bA_k(\bbeta)\|_{\max}\leq a_A,
 \qquad \|\bA_k(\bbeta)^{-1}\|_{\infty\to\infty}\leq M_A.
\]
If $\nu_k\geq M_Aa_A$, then the row-wise program in the main text yields
\[
 \|\widehat{\bTheta}_k(\bbeta)\|_{\infty\to\infty}\leq M_A,
 \qquad
 \|\bI-\widehat{\bTheta}_k(\bbeta)\widehat{\bA}_k(\bbeta)\|_{\infty,c_ss}
 \leq c_ss\max_k\nu_k.
\]
\end{proposition}

\noindent\textit{Proof.}
For row $j$, put $\boldsymbol\theta_j^0=\bA_k(\bbeta)^{-1}\boldsymbol e_j$. Symmetry and the inverse row-sum bound give $\|\boldsymbol\theta_j^0\|_1\leq M_A$, while
$\|(\widehat{\bA}_k-\bA_k)\boldsymbol\theta_j^0\|_\infty\leq a_AM_A\leq\nu_k$.
Thus $\boldsymbol\theta_j^0$ is feasible and the optimizer has $\ell_1$ norm at most $M_A$. Moreover, symmetry of $\widehat{\bA}_k$ makes row $j$ of $\bI-\widehat{\bTheta}_k\widehat{\bA}_k$ the transpose of $\boldsymbol e_j-\widehat{\bA}_k\widehat{\boldsymbol\theta}_{kj}$. Every entry is therefore bounded by $\nu_k$, and its action on a $c_ss$-sparse vector is at most $c_ss\nu_k$. $\square$

This proposition verifies the algebraic calibration part once a uniform Hessian approximation is available. It does not establish the calibrated-score event. In particular, bounded inverse row sums are structural and do not follow from sparsity alone \citep{cai2011clime}; likewise, proving $a_{n,h}=O(h^2)$ for empirical profiling and estimated calibration remains a separate statistical step.

\subsection{A location--scale example and its local curvature}\label{sup:template}

The assumptions are easiest to interpret under the following data-generating
model. At an honest site, let
\[
 Y_{ki}=a_k+\bx_{ki}^{\mathsf T}\bbeta^*+\sigma_k\epsilon_{ki},
 \qquad
 \bx_{ki}\sim N(\boldsymbol 0,\boldsymbol\Sigma_k),
 \qquad \epsilon_{ki}\perp\bx_{ki}.
\]
Take $\epsilon_{ki}$ to have a standard normal or Logistic distribution,
$0<\underline\sigma\leq\sigma_k\leq\overline\sigma<\infty$, and choose
central quantiles such as $(0.25,0.50,0.75)$. The site-specific intercepts are
then
\[
 b^*_{k\ell}=a_k+\sigma_kF_\epsilon^{-1}(\tau_\ell).
\]
Thus $a_k$, $\sigma_k$, and $\boldsymbol\Sigma_k$ may vary across sites while
the slope remains common. For example, one may take
$(\Sigma_k)_{rs}=\rho_k^{|r-s|}$ with $|\rho_k|\leq\rho_0<1$; its eigenvalues
and inverse row sums are then uniformly bounded, and its inverse is
row-sparse.

With comparable site sizes and independent observations, this model satisfies
the sampling, first-moment, and quantile-density conditions. To examine
curvature, fix $h>0$ and let $b_{k\ell,h}^*$ be the population profiled
intercept at $\bbeta^*$ for the smoothed loss. It may differ from the
unsmoothed quantile intercept $b_{k\ell}^*$. Write
\[
 K_h(u)=\frac{\exp(u/h)}{h\{1+\exp(u/h)\}^2},\qquad
 c_{k,h}=\frac1q\sum_{\ell=1}^q
 \mathbb E K_h(a_k+\sigma_k\epsilon_k-b_{k\ell,h}^*)>0.
\]
At $\bbeta^*$, independence and $\mathbb E\bx_k=0$ make the population
intercept--slope Hessian block zero. Hence the exact identity is
\[
 \bA_k(\bbeta^*)=c_{k,h}\boldsymbol\Sigma_k.
\]
This identity concerns the true slope. Away from it, the residual contains
$-\bx_k^{\mathsf T}(\bbeta-\bbeta^*)$, so the Hessian weights depend on
$\bx_k$. The profile Hessian is then generally not a scalar multiple of
$\boldsymbol\Sigma_k$, and sparsity of $\boldsymbol\Sigma_k^{-1}$ does
not ensure sparsity of every nearby inverse profile Hessian.

For fixed dimension and bandwidth, the local curvature bounds can instead
be obtained by a perturbation argument. Put $\bA_{k,0}=\bA_k(\bbeta^*)$
and suppose $\lambda_{\min}(\bA_{k,0})\geq c_0>0$,
$\lambda_{\max}(\bA_{k,0})\leq C_0$, and
$\|\bA_{k,0}^{-1}\|_{\infty\to\infty}\leq M_0$, uniformly over sites.
The covariance and scale bounds above provide such constants for fixed $h$.
On a neighborhood where
\[
 \sup_k\|\bA_k(\bbeta)-\bA_{k,0}\|_{2\to2}\leq c_0/2,\qquad
 \sup_k\|\bA_k(\bbeta)-\bA_{k,0}\|_{\infty\to\infty}
 \leq(2M_0)^{-1},
\]
the eigenvalues lie in $[c_0/2,C_0+c_0/2]$. Factoring
$\bA_k=\bA_{k,0}\{\bI+\bA_{k,0}^{-1}(\bA_k-\bA_{k,0})\}$
and using the convergent matrix geometric series also gives
$\|\bA_k^{-1}\|_{\infty\to\infty}\leq2M_0$.
Here $\|\cdot\|_{2\to2}$ is the spectral operator norm. Smooth densities,
Gaussian design moments, and the positive intercept Hessian give continuity
and bounded derivatives on a sufficiently small closed neighborhood for
fixed dimension and $h$. These yield the perturbation and local Lipschitz
bounds needed in Assumption~3.

This argument does not give dimension-free neighborhood sizes. In
$\mathcal C(r_0)$, $\|\bbeta-\bbeta^*\|_2\leq\sqrt{c_ss}\,r_0$;
thus the radius and curvature constants must be checked as $s,p$ grow or
$h$ decreases. The empirical calibration event in Assumption~4 needs
further uniform probability bounds. In particular, a Dantzig tolerance must
both cover the Hessian approximation error and keep $c_ss\max_k\nu_k$
small enough for contraction. The additional bound $a_{n,h}\lesssim h^2$
also requires control of the errors from empirical profiling and estimated
calibration; sample splitting alone does not prove it.

\section{Profile loss calculations}\label{sup:profile}

We first verify the local profile calculations. For a fixed site, suppress $k$ and write
\[
 L_h(\bbeta,\boldsymbol b)=\frac{1}{nq}
 \sum_{i=1}^n\sum_{\ell=1}^q
 \rho_{\tau_\ell,h}(Y_i-b_\ell-\bx_i^{\mathsf T}\bbeta).
\]
Let $\widehat{\boldsymbol b}(\bbeta)$ solve
$\nabla_{\boldsymbol b}L_h=0$. Because every $\tau_\ell\in(0,1)$, each intercept objective is coercive and strictly convex, so its minimizer is finite and unique. For logistic smoothing, the empirical intercept block is diagonal, with entries
\[
 \widehat D_{\ell\ell}=\frac{1}{nq}\sum_{i=1}^n
 \frac{\exp\{(Y_i-\widehat b_\ell-\bx_i^{\mathsf T}\bbeta)/h\}}
 {h[1+\exp\{(Y_i-\widehat b_\ell-\bx_i^{\mathsf T}\bbeta)/h\}]^2}>0.
\]
Thus it is nonsingular at every finite minimizer. Consequently, the implicit-function theorem gives
\[
 \frac{\partial\widehat{\boldsymbol b}(\bbeta)}
 {\partial\bbeta^{\mathsf T}}
 =-\{\nabla_{\boldsymbol b\boldsymbol b}^2L_h\}^{-1}
   \nabla_{\boldsymbol b\bbeta}^2L_h.
\]
Applying the chain rule to
$\widehat Q_h(\bbeta)=L_h\{\bbeta,\widehat{\boldsymbol b}(\bbeta)\}$
therefore yields the following familiar Schur complement.

\begin{lemma}[Profile derivatives]\label{sup:lem:profile}
Evaluate the full Hessian of $L_h$ at the fitted intercepts and candidate slope,
and order its coordinates as $(\boldsymbol b,\bbeta)$:
\[
 \begin{pmatrix}\widehat{\boldsymbol D}&\widehat{\boldsymbol C}\\
 \widehat{\boldsymbol C}^{\mathsf T}&\widehat{\boldsymbol B}\end{pmatrix},
\]
where the first block corresponds to $\boldsymbol b$. Then
\[
 \nabla\widehat Q_h(\bbeta)=
 \nabla_{\bbeta}L_h\{\bbeta,\widehat{\boldsymbol b}(\bbeta)\},
 \qquad
 \nabla^2\widehat Q_h(\bbeta)=
 \widehat{\boldsymbol B}-\widehat{\boldsymbol C}^{\mathsf T}
 \widehat{\boldsymbol D}^{-1}\widehat{\boldsymbol C}.
\]
\end{lemma}

\noindent\textit{Proof.}
The first identity follows because
$\nabla_{\boldsymbol b}L_h\{\bbeta,\widehat{\boldsymbol b}(\bbeta)\}=0$.
Differentiating that first-order condition gives the displayed derivative of
$\widehat{\boldsymbol b}(\bbeta)$. Substitution into the derivative of the
profile score proves the second identity. $\square$

At fixed population nuisance parameters, symmetry of the logistic smoothing kernel cancels the first-order convolution term. Specifically, if $F$ has a bounded density derivative and $K$ is the standard logistic density, Taylor expansion gives $\int F(u-hv)K(v)\,dv-F(u)=O(h^2)$, since $\int vK(v)\,dv=0$ and $\int v^2K(v)\,dv<\infty$. With suitable uniform moment and density bounds, this explains the population smoothing-bias scale. It does not show that $\mathbb E\widehat{\bg}_k(\bbeta^*)$ has the same order: empirical profiling uses estimated intercepts, and inverse calibration adds another data-dependent factor. We therefore retain the calibrated envelope $a_{n,h}$ in \eqref{sup:eq:scoredec}. Moreover, because derivatives of $h^{-1}K(\cdot/h)$ introduce inverse powers of $h$, any primitive verification must state bandwidth and sample-size conditions that keep $L_A$, $M_A$, and the Hessian approximation under control.

\section{Robust aggregation}\label{sup:aggregation}

\subsection{Proof and CQR example for Proposition 2}
Let $g_-=\min(g_1,g_2)$ and $g_+=\max(g_1,g_2)$. As the Byzantine value $v$ varies, $\operatorname{median}(g_1,g_2,v)$ ranges over exactly $[g_-,g_+]$. Its largest distance from $(g_1+g_2)/2$ is therefore $(g_+-g_-)/2$. Taylor's formula gives $g_k=a_ke+r_k$ with $|r_k|\leq L_ke^2/2$, which proves the lower bound. For $z_k=\beta-a_k^{-1}g_k$, we have $z_k-\beta^*=-r_k/a_k$. The median of two honest values and one arbitrary value always lies between the honest values, proving the calibrated bound. These are different reference quantities: the first measures score displacement and the second parameter error.

The assumptions hold in a scalar smoothed CQR model with $X$ uniform on $\{-1,1\}$ and $Y=\beta^*X+\sigma_k Z$, where $Z\sim N(0,1)$ is independent of $X$ and $\sigma_k>0$ may differ across sites. Fix $h>0$ and the quantile grid, and let $b_{k\ell,h}$ be the population fitted intercept at $\beta^*$. At this slope, residuals depend only on $Z$, so centering and independence give $Q_k'(\beta^*)=0$ and a zero intercept--slope Hessian block. With $K_h$ the logistic smoothing density,
\[
 a_k=\frac1q\sum_{\ell=1}^q
 \mathbb E K_h(\sigma_k Z-b_{k\ell,h})>0.
\]
Smoothness and positivity of the intercept Hessian give a smooth local profile by the implicit function theorem; on a sufficiently small compact interval its second derivative is Lipschitz. Unequal curvatures can be obtained by varying $\sigma_k$: for fixed $\sigma_1$, $a_1>0$, while convolution with the normal density yields $a_k\leq(\sigma_k\sqrt{2\pi})^{-1}\to0$ as $\sigma_k\to\infty$. Choose a finite $\sigma_2$ with $a_2<a_1$. This verifies the proposition for a common-slope CQR model at fixed bandwidth, without assuming that smoothing preserves stationarity in every model. It is a population example, not a verification of the full empirical event in Assumption 4.

\subsection{Contaminated trimmed mean}

We next isolate the only deterministic fact needed about the coordinatewise trimmed mean. Honest worker values may have different offsets and distributions; the argument uses only their average and their range.

\begin{lemma}[Contaminated trimmed mean]\label{sup:lem:trim}
Let $x_1,\ldots,x_m$ be real numbers, of which the values indexed by
$\Hcal$ are honest and the remaining $B\leq b$ values are arbitrary. Put
$b=\lfloor\gamma m\rfloor$, $q_m=m-2b>0$, and suppose the honest values admit
\[
 x_k=\theta+d_k+\xi_k,\qquad
 \max_{k\in\Hcal}|d_k|\leq D,\quad
 \left|m_H^{-1}\sum_{k\in\Hcal}\xi_k\right|\leq V,\quad
 \max_{k\in\Hcal}|\xi_k|\leq U.
\]
Then
\begin{equation}
 \left|\trim_\gamma(x_1,\ldots,x_m)-\theta\right|
 \leq \frac{m_H}{q_m}(D+V)+\frac{3b}{q_m}(D+U).             \label{sup:eq:trim}
\end{equation}
In particular, if $\gamma\leq\bar\gamma<1/2$, the right-hand side is at most
$C_{\bar\gamma}\{D+V+\gamma U/(1-2\gamma)\}$.
\end{lemma}

\noindent\textit{Proof.}
Let $R_H$ be the honest observations deleted from the two tails and let $S_B$
be the Byzantine observations that survive trimming. Any surviving Byzantine
value lies between the smallest and largest honest values. Indeed, a Byzantine
value strictly larger than every honest value would be among at most $B\leq b$
such upper-tail values and would therefore be deleted; the lower tail is
identical. Hence every element of $S_B$ is within $D+U$ of $\theta$.

Writing $T$ for the retained set, we have
\[
 \sum_{k\in T}(x_k-\theta)
 =\sum_{k\in\Hcal}(x_k-\theta)
  -\sum_{k\in R_H}(x_k-\theta)
  +\sum_{k\in S_B}(x_k-\theta).
\]
The first term is bounded in absolute value by $m_H(D+V)$. Moreover,
$|R_H|\leq2b$, $|S_B|\leq B\leq b$, and each summand in the last two terms is
bounded by $D+U$. Division by $q_m$ proves \eqref{sup:eq:trim}. Finally,
$q_m\geq m(1-2\gamma)$, $m_H\leq m$, and $b\leq\gamma m$, which gives the
simplified form. $\square$

Applying Lemma~\ref{sup:lem:trim} to each coordinate gives the vector version
on the joint coordinatewise event. Byzantine reports need no magnitude bound:
every retained value lies within the honest range for that coordinate.

\section{Expansion of a calibrated message}\label{sup:expansion}

Fix an honest worker and an iterate $\bbeta$. Set
$\boldsymbol e=\bbeta-\bbeta^*$ and
\[
 \overline{\bA}_k(\bbeta)=
 \int_0^1\widehat{\bA}_k(\bbeta^*+u\boldsymbol e)\,{\rm d}u.
\]
The fundamental theorem of calculus gives the exact identity
\[
 \widehat{\bg}_k(\bbeta)=\widehat{\bg}_k(\bbeta^*)
 +\overline{\bA}_k(\bbeta)\boldsymbol e.
\]
Consequently,
\begin{equation}
\begin{aligned}
 \bz_k(\bbeta)-\bbeta^*
={}&-\widehat{\bTheta}_k(\bbeta)\widehat{\bg}_k(\bbeta^*) \\
 &+\{\bI-\widehat{\bTheta}_k(\bbeta)
       \widehat{\bA}_k(\bbeta)\}\boldsymbol e \\
 &+\widehat{\bTheta}_k(\bbeta)
   \{\widehat{\bA}_k(\bbeta)-\overline{\bA}_k(\bbeta)\}
   \boldsymbol e.
\end{aligned}\label{sup:eq:newtonexpand}
\end{equation}
The second line is the calibration error. By
\eqref{sup:eq:lipschitz},
\[
 \|\widehat{\bA}_k(\bbeta)-\overline{\bA}_k(\bbeta)\|_{\infty\to\infty}
 \leq \int_0^1 L_A(1-u)\|\boldsymbol e\|_\infty\,{\rm d}u
 =\frac{L_A}{2}\|\boldsymbol e\|_\infty.
\]
Since $\bbeta\in\mathcal C(r_0)$, the vector $\boldsymbol e$ has at most
$c_ss$ nonzero coordinates and the integration segment stays in
$\mathcal C(r_0)$. Thus the restricted calibration norm applies to
$\boldsymbol e$. Combining this fact and the last inequality with
\eqref{sup:eq:cal}--\eqref{sup:eq:thetabound} proves the following result.

\begin{lemma}[Local Newton cancellation]\label{sup:lem:newton}
On the calibration and Lipschitz events, any honest message generated at an
iterate $\bbeta\in\mathcal C(r_0)$, with
$r=\|\bbeta-\bbeta^*\|_\infty$, satisfies
\begin{equation}
 \bz_k(\bbeta)-\bbeta^*
 =-\widehat{\bTheta}_k(\bbeta)\widehat{\bg}_k(\bbeta^*)
  +\boldsymbol R_k(\bbeta),\qquad
 \|\boldsymbol R_k(\bbeta)\|_\infty
 \leq\mu_n r+M_AL_A r^2.                                    \label{sup:eq:remainder}
\end{equation}
\end{lemma}

For uniformity over data-dependent iterates, the decomposition
\eqref{sup:eq:scoredec} and bounds
\eqref{sup:eq:pooled}--\eqref{sup:eq:max} are understood with
$\widehat{\bTheta}_k(\bbeta)$ in place of
$\widehat{\bTheta}_k(\bbeta^*)$, uniformly over
$\bbeta\in\mathcal C(r_0)$. We assume these uniform inequalities as stated in Section~\ref{sup:notation}; neither cross-fitting nor updating the preconditioner alone proves them.

\begin{theorem}[Uniform robust message bound]\label{sup:thm:message}
Suppose \eqref{sup:eq:cal}--\eqref{sup:eq:max} hold on $\mathcal E_\delta$, and let
$\alpha\leq\gamma\leq\bar\gamma<1/2$. There is a constant $C$, depending only
on $C_\xi,M_A,L_A,C_n$, and $\bar\gamma$, such that
\begin{align}
 &\left\|\trim_\gamma\{\bz_1(\bbeta),\ldots,\bz_m(\bbeta)\}
 -\bbeta^*\right\|_\infty \nonumber\\
 &\quad\leq C\left[
 a_{n,h}+\sqrt{\frac{\log(2p/\delta)}{N_H}}
 +\frac{\gamma}{1-2\gamma}
  \sqrt{\frac{\log(2pm/\delta)}{n_{\min}}}
 +\mu_n r+r^2\right]                                        \label{sup:eq:messagebound}
\end{align}
simultaneously for all $\bbeta\in\mathcal C(r_0)$ with
$r=\|\bbeta-\bbeta^*\|_\infty$.
\end{theorem}

\noindent\textit{Proof.}
For a fixed coordinate $j$, Lemma~\ref{sup:lem:newton} writes each honest
message as
\[
 z_{kj}=\beta_j^*+d_{kj,h}+\xi_{kj}+R_{kj}.
\]
Thus Lemma~\ref{sup:lem:trim} applies with
\[
 D=a_{n,h}+\mu_nr+M_AL_Ar^2,\quad
 V=C_\xi\sqrt{\frac{\log(2p/\delta)}{N_H}},\quad
 U=C_\xi\sqrt{\frac{\log(2pm/\delta)}{n_{\min}}}.
\]
For example, a sufficient explicit choice is
\[
 C=\max\left\{
 \frac{1+3\bar\gamma}{1-2\bar\gamma}\max(1,M_A L_A),\,
 \frac{C_\xi}{1-2\bar\gamma},\,3C_\xi\right\}.
\]
The factors $m_H/(m-2b)$ and $3b/(m-2b)$ in the trimming lemma
give this choice directly. It proves \eqref{sup:eq:messagebound} and makes
the dependence of the contraction conditions on the uniform envelopes explicit.
The events were already
constructed uniformly over coordinates and the sparse neighborhood, so no
additional union bound is needed at this step. $\square$

\begin{remark}
Without comparable site sizes, the same proof remains valid after replacing
$N_H^{-1/2}$ in \eqref{sup:eq:messagebound} by
$m_H^{-1}(\sum_{k\in\Hcal}n_k^{-1})^{1/2}$, provided the assumed uniform
honest-average bound is replaced by that scale as well. Site-size comparability
alone does not establish this probability bound for estimated calibrated scores.
\end{remark}

\begin{lemma}[Worker concentration]\label{sup:lem:concentration}
Suppose the vectors $\boldsymbol\xi_k$, $k\in\Hcal$, are independent and
mean zero, and each coordinate has sub-Gaussian norm at most
$K/\sqrt{n_k}$. If $n_{\max}/n_{\min}\leq C_n$, then, with probability at
least $1-\delta$,
\begin{align*}
 \left\|m_H^{-1}\sum_{k\in\Hcal}\boldsymbol\xi_k\right\|_\infty
 &\leq C K\sqrt{\frac{\log(4p/\delta)}{N_H}},\\
 \max_{k\in\Hcal}\|\boldsymbol\xi_k\|_\infty
 &\leq C K\sqrt{\frac{\log(4pm/\delta)}{n_{\min}}},
\end{align*}
where $C$ depends only on $C_n$.
\end{lemma}

\noindent\textit{Proof.}
For each coordinate, the unweighted average is sub-Gaussian with variance
proxy
\[
 \frac{K^2}{m_H^2}\sum_{k\in\Hcal}\frac{1}{n_k}
 \leq \frac{C(C_n)K^2}{N_H}.
\]
A union bound over $p$ coordinates gives the first result. For the second,
the sub-Gaussian tail at worker $k$ is bounded by
$2\exp(-c n_k t^2/K^2)$. A union bound over at most $pm$ worker-coordinate
pairs and the choice
$t=CK\{\log(4pm/\delta)/n_{\min}\}^{1/2}$ complete the proof. $\square$

This fixed-vector lemma illustrates the two stochastic scales in the main theorem; it is not a uniform empirical-process proof for the implemented estimator. Under its additional hypotheses, the
calibrated worker score is of local order $n_k^{-1/2}$, whereas its honest
average is of pooled order $N_H^{-1/2}$. Trimming replaces a fraction of that
average by values bounded only through the honest worker range, which accounts
for the additional factor $\gamma$ multiplying the local scale.

\section{Iterative thresholding and support recovery}\label{sup:iteration}

For the thresholding argument, let
\[
 \epsilon=a_{n,h}+\sqrt{\frac{\log(2p/\delta)}{N_H}}+
 \frac{\gamma}{1-2\gamma}
 \sqrt{\frac{\log(2pm/\delta)}{n_{\min}}},
 \qquad B(r)=C(\epsilon+\mu_nr+r^2),
\]
where $C$ is the constant in Theorem~\ref{sup:thm:message}. At round $t$,
write $\widetilde{\bz}^{(t)}$ for the coordinatewise trimmed mean and update
$\bbeta^{(t+1)}=\mathcal S_{\lambda_t}(\widetilde{\bz}^{(t)})$.

\begin{lemma}[Soft-thresholding]\label{sup:lem:soft}
If $\|\boldsymbol a-\bbeta^*\|_\infty\leq\lambda$, then
\[
 \supp\{\mathcal S_\lambda(\boldsymbol a)\}\subseteq\Scal,\qquad
 \|\mathcal S_\lambda(\boldsymbol a)-\bbeta^*\|_\infty\leq2\lambda.
\]
Consequently, its $\ell_2$ and $\ell_1$ errors are bounded by
$2\sqrt{s}\lambda$ and $2s\lambda$, respectively.
\end{lemma}

\noindent\textit{Proof.}
For $j\notin\Scal$, $|a_j|\leq\lambda$, so the output coordinate is zero.
For $j\in\Scal$, the contraction property of scalar soft-thresholding gives
\[
 |\mathcal S_\lambda(a_j)-\beta_j^*|
 \leq |a_j-\beta_j^*|+\lambda\leq2\lambda.
\]
The support inclusion then converts the sup-norm bound into the stated
$\ell_2$ and $\ell_1$ bounds. $\square$

\begin{lemma}[Median-thresholded initialization]\label{sup:lem:init}
Let $\boldsymbol u_1,\ldots,\boldsymbol u_m$ be the reported local initial estimates, with fewer than $m/2$ Byzantine reports. If
\[
 \max_{k\in\Hcal}\|\boldsymbol u_k-\bbeta^*\|_\infty\leq a_0,
\]
then their coordinatewise median $\widetilde{\bbeta}^{\rm loc}$ satisfies
$\|\widetilde{\bbeta}^{\rm loc}-\bbeta^*\|_\infty\leq a_0$. Consequently, for $\lambda_{\rm init}\geq a_0$,
\[
 \supp\{\mathcal S_{\lambda_{\rm init}}(\widetilde{\bbeta}^{\rm loc})\}
 \subseteq\Scal,\qquad
 \|\mathcal S_{\lambda_{\rm init}}(\widetilde{\bbeta}^{\rm loc})-\bbeta^*\|_\infty
 \leq2\lambda_{\rm init}.
\]
In particular, this initializer belongs to $\mathcal C(r_0)$ if $2\lambda_{\rm init}\leq r_0$.
\end{lemma}

\noindent\textit{Proof.}
For every coordinate, more than half of the reports lie in
$[\beta_j^*-a_0,\beta_j^*+a_0]$, so every coordinatewise median lies in that interval. Lemma~\ref{sup:lem:soft} gives the remaining claims; support inclusion also gives $|\supp(\bbeta^{(0)})\cup\Scal|=s\leq c_ss$. $\square$

\begin{theorem}[Error recursion]\label{sup:thm:recursion}
Fix a finite horizon $T$ and deterministic radii $\bar r_t\geq0$. On $\mathcal E_\delta$, if, for every $t=0,\ldots,T-1$,
$\|\bbeta^{(t)}-\bbeta^*\|_\infty\leq\bar r_t$ and
$\bbeta^{(t)}\in\mathcal C(r_0)$, then choosing
$\lambda_t\geq B(\bar r_t)$ yields, simultaneously over these rounds,
\[
 \supp(\bbeta^{(t+1)})\subseteq\Scal,\quad
 \|\bbeta^{(t+1)}-\bbeta^*\|_\infty\leq2\lambda_t,
\]
together with $\ell_2$ and $\ell_1$ bounds
$2\sqrt{s}\lambda_t$ and $2s\lambda_t$.
\end{theorem}

\noindent\textit{Proof.}
By Theorem~\ref{sup:thm:message},
$\|\widetilde{\bz}^{(t)}-\bbeta^*\|_\infty
\leq B(\|\bbeta^{(t)}-\bbeta^*\|_\infty)\leq B(\bar r_t)\leq\lambda_t$.
Lemma~\ref{sup:lem:soft} then gives every conclusion. Since the event is
uniform over the sparse neighborhood, it remains valid for iterates computed
from the same data in earlier rounds. $\square$

\begin{corollary}[Terminal rate]\label{sup:cor:terminal}
Suppose $\mathcal E_\delta$ holds and the initializer belongs to $\mathcal C(r_0)$ with $r_0>0$ and $c_s\geq1$. Define $\bar r_0=r_0$, $\lambda_t=B(\bar r_t)$, and $\bar r_{t+1}=2B(\bar r_t)$. If
$C\mu_n\leq1/8$, $Cr_0\leq1/8$, and $\epsilon$ is small enough that
$4C\epsilon\leq r_0$, then
\[
 \|\bbeta^{(t)}-\bbeta^*\|_\infty\leq\bar r_t
 \leq 4C\epsilon+2^{-t}r_0.
\]
Thus $T=\max\{1,\lceil\log_2\{r_0/(4C\epsilon)\}\rceil\}$ rounds suffice for
$\|\bbeta^{(T)}-\bbeta^*\|_\infty\leq8C\epsilon$, which is $O(\epsilon)$
when $C$ is uniformly bounded. If
$\min_{j\in\Scal}|\beta_j^*|>2B(\bar r_{T-1})$, exact support recovery holds.
\end{corollary}

\noindent\textit{Proof.}
Whenever $\bar r_t\leq r_0$,
\[
 \bar r_{t+1}=2C(\epsilon+\mu_n\bar r_t+\bar r_t^2)
 \leq2C\epsilon+\tfrac12 \bar r_t,
\]
because $2C\mu_n\leq1/4$ and $2C\bar r_t\leq1/4$. Iteration of this affine
recursion gives
$\bar r_t\leq4C\epsilon+2^{-t}r_0$. To justify applying the recursion at every step, note separately that $2C\epsilon+r_0/2\leq r_0$. Induction gives $\bar r_t\leq r_0$ and $\|\bbeta^{(t)}-\bbeta^*\|_\infty\leq\bar r_t$, while the support inclusion after each update ensures the sparse-neighborhood condition. Lemma~\ref{sup:lem:init} supplies this initializer on the stated honest-local estimation event.

For support recovery, false positives are already excluded by
Theorem~\ref{sup:thm:recursion}. Furthermore,
$|\widetilde z_j^{(T-1)}|\geq|\beta_j^*|-B(\bar r_{T-1})>
B(\bar r_{T-1})=\lambda_{T-1}$ for every $j\in\Scal$. Hence none of the true
coordinates is thresholded to zero. $\square$

When empirical calibration is exact ($\mu_n=0$), the deterministic recursion reduces to $\bar r_{t+1}=2C\bar r_t^2+2C\epsilon$. This removes the linear calibration remainder. The logarithmic bound above remains sufficient; no sharper uniform round count is needed for the stated result.

\section{Identification of an honest-site slope average}\label{sup:identification}

Fix $m\geq3$ and an integer $1\leq b<m/2$, and write $m_H=m-b$. Partition the workers into a set $S$ of size $m_H-1$, two individual workers $u,v$, and a set $R$ of size $b-1$. Reports at $S$ and $R$ are zero; reports at $u$ and $v$ are $-\boldsymbol a$ and $\boldsymbol a$, respectively. Configuration A declares $S\cup\{u\}$ honest, whereas configuration B declares $S\cup\{v\}$ honest. All other workers are Byzantine. Thus both configurations have exactly $b$ faulty workers, identical ordered reports, and honest means $-\boldsymbol a/m_H$ and $\boldsymbol a/m_H$.

The same construction applies even if each worker reports its entire data-generating law rather than only a slope. At worker $k$ with reported slope $\boldsymbol v_k$, let the reported law be $\bx\sim N_p(\boldsymbol0,\bI)$ and $Y=\bx^{\mathsf T}\boldsymbol v_k+E$, with $E\sim N(0,1)$ independent. Honest workers follow that law; a Byzantine worker can mimic it exactly. The observed laws, and hence the laws of any finite samples from them, coincide across configurations. The target separation is $2\|\boldsymbol a\|_2/m_H$. By the triangle inequality, a deterministic estimator has error at least half this separation in one configuration. For a randomized estimator, the same lower bound holds for the larger of its two expected errors, because its output distribution is identical in the two configurations.

No replication argument is needed. Under a fraction upper bound $\alpha<1/2$, choose any admissible integer $b\geq1$ with $b\leq\lfloor\alpha m\rfloor$. If $\lfloor\alpha m\rfloor=0$, this construction is unavailable and does not imply nonidentification. The proposition concerns an average of arbitrary honest-site slopes, not the common-slope CQR parameter studied elsewhere in the paper.

\section{Complete simulation specification}\label{sup:simulation}

This section records the data-generating mechanism used in the reported simulations.
For replication $r$ and setting index $a$, the seed was
$20260907+1000a+r$, with $r=0,\ldots,49$. There were $m=20$ sites, $n=90$
observations per site, and $p=30$ predictors. Let
\[
 \phi_k=\frac{2\pi(k-1)}{m},\qquad k=1,\ldots,m,\qquad
 d_{kj}=\exp\{0.45H\sin(\phi_k+(j-1)/4)\},
\]
and
\[
 \rho_k=0.05+0.50H\{0.5+0.5\cos(\phi_k)\}.
\]
For observation $i$, draw $Z_{ki}\sim N(0,1)$ and
$\boldsymbol\epsilon_{ki}\sim N_p(\boldsymbol0,\bI_p)$ independently, and
form the preliminary covariate
\[
 X^{\rm pre}_{kij}=d_{kj}\left[
 \sqrt{\max(1-\rho_k^2,0.05)}\,\epsilon_{kij}+\rho_kZ_{ki}\right]
 +0.35H\sin(\phi_k)a_j,
\]
where $a_j$ is the $j$th entry of the equally spaced vector from $-1$ to
$1$. Every column was then centered and scaled by its site-specific empirical
standard deviation, with a lower scale bound of $0.25$. Hence heterogeneity in
the correlation pattern remains even after marginal standardization.

The response was generated from
\[
 Y_{ki}=\ell_k+\bx_{ki}^{\mathsf T}\bbeta^*+\sigma_k E_{ki},
\]
where $E_{ki}=T_{ki}/\sqrt3$, $T_{ki}\sim t_3$, and
\[
 \ell_k=0.9H\cos(\phi_k-0.4),\qquad
 \sigma_k=0.65+0.75H\{0.5+0.5\sin(\phi_k+0.7)\}.
\]
Thus the $t_3$ error has unit variance before multiplication by $\sigma_k$,
while both its location and scale differ over sites. The slope was
\[
 \bbeta^*=(1.00,-0.82,0.66,-0.46,0.24,0,\ldots,0)^{\mathsf T}.
\]

For a Byzantine fraction $\alpha$, the last
$\lfloor\alpha m\rfloor$ worker positions were
treated as malicious, and all received messages were shuffled before
aggregation. If $\boldsymbol c$ and $\boldsymbol s_c$ denote the
coordinatewise median and empirical standard deviation of the honest
messages, respectively, with every standard deviation floored at $10^{-3}$ as in the code, the sign-reversal attack sent
$-4\boldsymbol c$, using the sign-flipping construction described in the
main text. The mimic attack used the center-preserving, heterogeneity-aware
construction described there and sent
\[
 \boldsymbol c+0.9\boldsymbol s_c\odot
 \operatorname{sign}(\overline{\bz}_{H}-\boldsymbol c).
\]
Zero signs were replaced independently by $-1$ or $1$. With more than one
malicious worker, independent $N(0,0.02^2s_{c,j}^2)$ perturbations avoided
exact duplicate vectors.
Both attack rules are constructed after observing the current
honest messages and therefore represent omniscient stress tests.

The quantile levels were $(0.25,0.50,0.75)$. HC-RCQR used bandwidth $0.24$,
local penalty $0.24\sqrt{\log(p)/n}$, ridge $0.08$, four communication rounds,
server threshold floor $0.035+0.035\alpha$, and trim fraction
$\min\{\max(\alpha+0.05,0.05),0.40\}$. The threshold at round $t$ was the
threshold floor multiplied by $1+0.5^{t+1}$, with $t=0,1,2,3$; initialization used $1.5$ times the floor. The raw-score comparator started at zero and used 28 steps of size $0.25$. The Oracle penalty was $0.012$. Med-L used hard thresholding, whereas HC-RCQR used soft thresholding. Reported support statistics for all methods used the evaluation set $\{j:|\widehat\beta_j|>0.08\}$, which need not equal the nonzero support in the theoretical result. Empirical site standardization also induces within-site covariate dependence; this design is a numerical illustration, not a verification of the independent-sampling assumptions.

Table~\ref{sup:tab:alpha} gives the complete Byzantine-fraction slice
under the mimic attack. The increase from $\alpha=0$ to $\alpha=0.2$ was
modest for HC-RCQR, whereas RTM-score deteriorated steadily. Because the
mimic values were not extreme, Mean and Med-L happened to be close in this
particular slice; neither matched the calibrated estimator.

\begin{table}[t]
\caption{Effect of the Byzantine fraction under $H=0.8$ and the
heterogeneity-mimic attack. Error is mean (Monte Carlo standard error) over
50 replications}\label{sup:tab:alpha}
\centering
\small
\begin{tabular}{clrr}
\toprule
$\alpha$ & Method & $\ell_2$ error & Exact support \\
\midrule
0 & Oracle & 0.078 (0.003) & 1.00 \\
  & Mean & 0.314 (0.004) & 0.96 \\
  & Med-L & 0.302 (0.004) & 0.92 \\
  & RTM-score & 0.462 (0.004) & 0.96 \\
  & HC-RCQR & 0.096 (0.003) & 1.00 \\
\addlinespace
0.1 & Oracle & 0.076 (0.002) & 1.00 \\
  & Mean & 0.311 (0.003) & 0.98 \\
  & Med-L & 0.302 (0.003) & 0.96 \\
  & RTM-score & 0.479 (0.004) & 0.86 \\
  & HC-RCQR & 0.101 (0.003) & 1.00 \\
\addlinespace
0.2 & Oracle & 0.078 (0.002) & 1.00 \\
  & Mean & 0.323 (0.004) & 0.96 \\
  & Med-L & 0.321 (0.005) & 0.86 \\
  & RTM-score & 0.527 (0.005) & 0.56 \\
  & HC-RCQR & 0.116 (0.004) & 1.00 \\
\bottomrule
\end{tabular}
\end{table}

\section{Bike Sharing analysis details}\label{sup:bike}

We used the 17,379 records in the hourly Bike Sharing file from the UCI Machine
Learning Repository \citep{fanaee2013bike}. The response was
$\log(1+\mathtt{cnt})$. We defined site
$12\,\mathtt{yr}+\mathtt{mnth}-1$, producing 24 monthly sites, and ordered
each site by the original \texttt{instant} index. The first
$\lfloor0.8n_k\rfloor$ records were used for training and the remainder for
testing.

The predictors were temperature, apparent temperature, humidity, wind speed,
holiday, working day, 23 hour indicators, six weekday indicators, and three
weather indicators, for a total of 38. The first four variables were
standardized with means and standard deviations computed only from the pooled
training observations. Neither \texttt{casual} nor \texttt{registered} was
used because those two counts add to the response count. Raw dates, year,
month, season, and the record identifier were also excluded.

HC-RCQR used bandwidth $0.18$, five refinement rounds, ridge $0.055$, local
penalty $0.20\sqrt{\log(p)/\operatorname{median}_{k\in\Hcal} n_k}$, server threshold floor
$0.018+0.030\alpha$, and trim fraction
$\min\{\max(\alpha+0.05,0.05),0.40\}$. RTM-score used 35 steps of size
$0.18$, while the Oracle penalty was $0.004$. For each predicted site and
method, quantile intercepts were the $(0.25,0.50,0.75)$ quantiles of that
site's training residuals. As in the simulation, the initialization threshold is $1.5$ times the floor and refinement threshold $t$ is $(1+0.5^{t+1})$ times the floor. The use of honest-site sample sizes for the penalty is a benchmark convenience based on known attack assignments; deployment requires a rule that does not use unknown membership. Pooled training standardization is trusted offline preprocessing, and only messages are attacked. Since training records from later months can occur after test records from earlier months, this is within-site held-out evaluation, not a global chronological forecasting experiment. Serial dependence and the common quantile-slope assumption have not been verified for these data.

The centered monthly training design matrices have numerical ranks between
35 and 37, below the 38 columns used for prediction. Thus empirical profile
Hessians are singular even though the local sample sizes exceed the dimension.
Ridge inversion makes these numerical updates well defined; this calculation
does not verify the population curvature or contraction conditions.

Finally, with four Byzantine sites, the ten assignment seeds were
$20261001+100a+r$, where $a$ indexes the attack setting and
$r=0,\ldots,9$. Table~\ref{sup:tab:width} reports the interval-width
comparison omitted from the main table. The near-nominal coverage reported
in the main text should therefore be interpreted together with width:
HC-RCQR's intervals were markedly shorter than those of the other distributed
methods.

\begin{table}[t]
\caption{Mean width (standard error) of the fitted interquartile prediction
interval on honest Bike Sharing test sites}\label{sup:tab:width}
\centering
\small
\begin{tabular}{lrrr}
\toprule
Method & No attack & Sign reversal & Heterogeneity-mimic \\
\midrule
Oracle & 0.732 (0.000) & 0.732 (0.002) & 0.731 (0.003) \\
Mean & 1.228 (0.000) & 1.616 (0.003) & 1.232 (0.003) \\
Med-L & 1.225 (0.000) & 1.257 (0.004) & 1.237 (0.004) \\
RTM-score & 1.396 (0.000) & 1.430 (0.005) & 1.411 (0.003) \\
HC-RCQR & 0.824 (0.000) & 0.857 (0.002) & 0.834 (0.003) \\
\bottomrule
\end{tabular}
\end{table}

The Oracle objective uses honest records only, but its numerical solver starts from the HC-RCQR estimate. It uses at most 180 iterations in the simulation and 220 in the application, with stopping tolerance $2\times10^{-6}$. Thus finite-iteration Oracle results may depend on the starting value; convergence from independent starts was not checked.

\section{Reproducibility notes}\label{sup:repro}

The supplied archive was executed with Python 3.12.14, NumPy 2.3.5, and
pandas 3.0.1. The core estimator depends only on NumPy; pandas is used to
read and encode the application data. The following commands, run from the
\texttt{code} directory after placing the UCI \texttt{hour.csv} file in the
documented data directory, reproduce the reported files:
\begin{verbatim}
python run_simulation.py --reps 50 --output ../results
python run_bikesharing.py --data ../data/hour.csv \
    --reps 10 --output ../results
python make_artifacts.py
\end{verbatim}
The ridge implementation caps each Newton step at Euclidean radius
$2.5\max\{1,\|\bbeta\|_2\}$. If $\eta_k\in(0,1]$ is the factor applied
to an honest step, its capped message is
$\bz_{k,\eta}=\bbeta-\eta_k\widehat{\bTheta}_k\widehat{\bg}_k(\bbeta)$.
Writing $\boldsymbol e=\bbeta-\bbeta^*$ and using the uncapped message
$\bz_k$ gives the exact relation
\[
 \bz_{k,\eta}-\bbeta^*
 =(1-\eta_k)\boldsymbol e+\eta_k(\bz_k-\bbeta^*).
\]
When the cap is active, an extra first-order error remains, and the
site-dependent factors also change the pooled score term. The present
proof therefore covers uncapped messages. The original experiments
did not record cap activations; the added paired experiment does. The preset numerical thresholds have likewise not
been proved to exceed the analytical error bounds $B(\bar r_t)$.
The original comparisons use different initializations and iteration budgets.
Section~\ref{sup:matched} adds a comparison with common initialization,
thresholds and iteration counts; it records cap activations and paired errors.

Every summary table is generated from a replicate-level comma-separated file.
No displayed entry was transcribed from console output. The UCI source data
are not included in the archive; the README gives the permanent DOI and the
expected file name.

\section[A finite-dimensional probability verification]{A finite-dimensional probability verification}\label{sup:verification}
This section verifies the calibration and pilot requirements in a bounded
scalar CQR model. It uses the same observations for all local fits. The result
does not cover growing dimension, shrinking bandwidth, or numerical solver
error. In particular, it is not a verification for the simulation design.

\begin{theorem}[A scalar CQR model]\label{sup:thm:primitive}
Fix $q$, $\eta\in(0,1/2)$, quantiles $\tau_\ell\in[\eta,1-\eta]$,
bandwidth $h_0>0$, and $0<\sigma_-\leq\sigma_+<\infty$.
Assume $|\mathcal B|\leq\alpha m$ and
$\alpha\leq\gamma\leq\bar\gamma<1/2$; Byzantine reports are unrestricted.
Each honest site has $n$ independent observations
\[
 X_{ki}\sim\operatorname{Unif}\{-1,1\},\qquad
 Y_{ki}=a_k+X_{ki}\beta^*+\sigma_k U_{ki},\qquad
 U_{ki}\sim\operatorname{Unif}[-1,1],
\]
where $X_{ki}$ and $U_{ki}$ are independent, sites are independent,
$\sigma_k\in[\sigma_-,\sigma_+]$, and $a_k\in\mathbb R$ are arbitrary.
Let $\beta^*\ne0$, use exact profile minimization and the scalar ridge inverse
$\widehat\Theta_k(\beta)=\{\widehat A_k(\beta)+\kappa\}^{-1}$,
with $0\leq\kappa\leq1$. For any fixed $R>0$, there are positive constants
$c,K$ depending only on $R,h_0,q,\eta,\sigma_-,\sigma_+$ such that the following
holds. Put $L=\log(8m/\delta)$ and suppose $n\geq K L$.
With probability at least $1-\delta$, uniformly over honest sites and
$|\beta-\beta^*|\leq R$, Assumptions~\ref{ass:sampling}--\ref{ass:calibration} of the main text hold with
$p=s=c_s=1$, bounded curvature and row-norm constants,
\[
 \mu_n\leq K\kappa,\qquad a_{n,h_0}=K L/n,
\]
and a score constant $C_\xi\leq K$. The optional condition
$a_{n,h}\leq C_bh^2$ is not used. An exact minimizer of
$\widehat Q_k(\beta)+\lambda_{\rm loc}|\beta|$, with
$\lambda_{\rm loc}=\sqrt{L/n}$, obeys simultaneously
\[
 \max_{k\in\mathcal H}|\widehat\beta_k^{\rm loc}-\beta^*|
 \leq K\sqrt{L/n}.
\]
Consequently, for fixed $\gamma\leq\bar\gamma<1/2$ covering the Byzantine
fraction, sufficiently small fixed $\kappa$, and sufficiently large $n/L$,
median initialization with $\lambda_{\rm init}=K\sqrt{L/n}$ and the
analytical threshold schedule of Corollary 1 give, with probability at least
$1-\delta$, an uncapped terminal estimate satisfying
\[
 |\widehat\beta-\beta^*|\leq K\left\{
 \frac{L}{n}+\sqrt{\frac{\log(2/\delta)}{m_H n}}
 +\frac{\gamma}{1-2\gamma}\sqrt{\frac{\log(2m/\delta)}{n}}\right\}.
\]
Constants can be enlarged between displays. The required neighborhood and
threshold choices are justified below; they are not the preset numerical tunings.
\end{theorem}

\noindent\textit{Proof.}
Write $\varepsilon_k=\sigma_kU_k$ and center each intercept by $a_k$.
The true quantile intercept is $a_k+\sigma_k(2\tau_\ell-1)$.
Thus the common-slope model and the local density conditions hold; the density
equals $1/(2\sigma_k)$ in a common neighborhood of the selected quantiles.
The bounded design and errors also imply the sampling and moment conditions.

\textit{Uniform concentration used in the proof.}
We first record the elementary bound needed below. A uniformly bounded,
uniformly Lipschitz family $f_t$ indexed by a fixed-dimensional compact
rectangle satisfies
\[
 \sup_t|(P_n-P)f_t|\leq K\sqrt{\log(2/\epsilon)/n}
 \quad\hbox{with probability at least }1-\epsilon.
\]
The constants depend on the envelope, Lipschitz constant, dimension and
rectangle. To see this, cover the rectangle by grids of mesh $2^{-j}$.
The number of adjacent projection pairs at level $j$ is at most $K2^{Kd j}$
for a fixed dimension $d$. Hoeffding's inequality applied to increments,
whose envelope is at most $K2^{-j}$, and a union bound with failure budgets
$\epsilon 2^{-j-1}$ give increment bounds
$K2^{-j}\sqrt{\{j+\log(2/\epsilon)\}/n}$.
Summing over levels, treating the initial grid separately, and using
continuity proves the claim. The same proof applies to independent,
non-identically distributed summands with common envelopes.

\textit{Profile derivatives and local concentration.}
Put $e=\beta-\beta^*$ and
$K_{h_0}(u)=\psi'_{\tau,h_0}(u)$; this derivative does not depend on $\tau$.
For $|e|\leq R$, each population or empirical fitted centered intercept lies
in a fixed compact interval. Indeed, $\varepsilon-Xe$ is bounded by
$\sigma_++R$, and shifting an intercept beyond this interval by
$h_0\log\{(1-\eta)/\eta\}+1$ makes its score have a fixed sign.
The logistic density and all derivatives needed below are bounded on this
compact residual interval, and $K_{h_0}$ has a positive lower bound there.
The population intercept Hessian is therefore bounded away from zero.

Apply the concentration bound to the full loss derivatives through order
three over the compact set of $(e,b_1,\ldots,b_q)$, and to the two design
counts, using a union bound over at most $m$ honest sites. Except on an event
of probability at most $\delta/2$, both design values occur with proportions
at least $1/4$, and all these derivative deviations are at most
$K t_n$, where $t_n=\sqrt{L/n}$.
Monotonicity of the intercept score and its positive derivative then give
uniform fitted-intercept errors at most $K t_n$.

For each quantile, the scalar Schur complement is the weighted variance of
$X$ times its mean weight, divided by $q$. Because both design proportions
are bounded below and every weight has a positive lower bound, the empirical
and population profile Hessians are bounded below by a fixed $c>0$ on this
event. The implicit derivative formula for the fitted intercepts and the
bounded third derivatives give a common Lipschitz constant for these profile
Hessians. The Schur formula also yields
\[
 \sup_{k,|e|\leq R}|\widehat A_k(\beta)-A_k(\beta)|\leq Kt_n.
\]
This proves the population curvature requirements and empirical Lipschitz
bound. Since the Hessians are scalar and bounded below, both ridge inverses
are bounded uniformly even at $\kappa=0$. In particular,
$|1-\widehat\Theta_k\widehat A_k|=\kappa\widehat\Theta_k\leq K\kappa$.

\textit{Same-sample score expansion.}
Let $c_{k\ell}$ be the population smoothed intercept at $\beta^*$, centered
by $a_k$. Define independent, mean-zero site scores
\[
 S_k=-\frac1n\sum_{i=1}^n X_{ki}\frac1q\sum_{\ell=1}^q
 \psi_{\tau_\ell,h_0}(\varepsilon_{ki}-c_{k\ell}).
\]
Independence of $X$ and $\varepsilon$ and $\mathbb EX=0$ imply both
$\mathbb ES_k=0$ and
$\mathbb E[XK_{h_0}(\varepsilon-c_{k\ell})]=0$.
Taylor expansion in the empirical fitted intercept gives
\[
 \widehat g_k(\beta^*)=S_k+O(t_n^2)
 \quad\hbox{uniformly over honest }k.
\]
Specifically, the linear term is a fitted-intercept error of order $t_n$
times $n^{-1}\sum_iX_{ki}K_{h_0}(\varepsilon_{ki}-c_{k\ell})=O(t_n)$;
the quadratic term is bounded by the second derivative times the squared
intercept error. This bound uses no independence between these two empirical
quantities. It is precisely the same-sample remainder.

Set $\Theta_k^0(\beta)=\{A_k(\beta)+\kappa\}^{-1}$.
The Hessian bound implies
$\sup_{k,\beta}|\widehat\Theta_k(\beta)-\Theta_k^0(\beta)|\leq Kt_n$,
and Hoeffding gives $\max_k|S_k|\leq Kt_n$. Hence
\[
 -\widehat\Theta_k(\beta)\widehat g_k(\beta^*)
 =\underbrace{-\Theta_k^0(\beta)S_k}_{\xi_k(\beta)}+d_k(\beta),
 \qquad \sup_{k,\beta}|d_k(\beta)|\leq K L/n.
\]
The random remainder $d_k$ is allowed by Assumption 4, which bounds an offset
without requiring it to be an expectation. This calculation does not assert
that the implemented calibrated score is unbiased.

The deterministic functions $\Theta_k^0(\beta)$ are uniformly bounded and
Lipschitz on the interval. Apply the same concentration argument to the
$m_Hn$ independent mean-zero summands defining
$m_H^{-1}\sum_k\Theta_k^0(\beta)S_k$, now indexed only by $\beta$.
With a further failure budget $\delta/2$, including the local score maxima,
\[
 \sup_\beta\left|m_H^{-1}\sum_k\xi_k(\beta)\right|
 \leq K\sqrt{\frac{\log(2/\delta)}{m_Hn}},\qquad
 \sup_{k,\beta}|\xi_k(\beta)|
 \leq K\sqrt{\frac{\log(2m/\delta)}{n}}.
\]
Allocation of the finitely many failure budgets only changes constants.
These are the required uniform bounds with $a_{n,h_0}=K L/n$.
At the population slope the score is exactly zero for every fixed $h_0$,
by independence and centering; thus there is no population smoothing shift
in the slope in this particular model. This does not remove its finite-sample
offset $K L/n$.

\textit{Pilot and compatible parameters.}
On the same event $|\widehat g_k(\beta^*)|\leq Kt_n$.
For $|e|\leq R$, strong convexity and the triangle inequality give
\[
 \widehat Q_k(\beta^*+e)+\lambda_{\rm loc}|\beta^*+e|
 -\widehat Q_k(\beta^*)-\lambda_{\rm loc}|\beta^*|
 \geq \frac c2 e^2-(Kt_n+\lambda_{\rm loc})|e|.
\]
Choose a constant $K_0$ with $cK_0/2>K+1$.
For $n/L$ large enough, $K_0t_n<R$ and the penalized objective exceeds its
value at $\beta^*$ at both endpoints $\beta^*\pm K_0t_n$.
Convexity therefore confines every minimizer to that interval. A minimizer
exists: with both design signs present and quantiles away from the endpoints,
the profiled loss is coercive. The median of more than half accurate pilots
has the same error bound, and soft thresholding at $\lambda_{\rm init}=K_0t_n$
gives error at most $2K_0t_n$.

Fix the derivative and message constants on the interval of radius $R$.
Choose $0<r_0\leq R$ with $Cr_0\leq1/8$, and then $\kappa$ with
$CK\kappa\leq1/8$. Finally choose $n/L$ large enough that
$2K_0t_n\leq r_0$ and $4C\varepsilon\leq r_0$, where $\varepsilon$ is the
three-term rate in the theorem. These choices are possible for
$\log(m/\delta)=o(n)$ and fixed $\bar\gamma<1/2$.
The deterministic trimming, thresholding and induction arguments in
Appendices~\ref{sup:aggregation}--\ref{sup:iteration} now give the displayed terminal bound on the joint event.
Support recovery follows under the corresponding beta-min condition; with
one nonzero coefficient it does not address selection among many predictors.
\hfill$\square$

This example supplies a primitive verification for fixed-dimensional CQR,
including multiple quantiles, but leaves the general sparse model conditional.
It also exposes a local $L/n$ term that should not be discarded when many sites
are pooled. For example, for fixed $\delta$ this term is negligible relative
to $(m_Hn)^{-1/2}$ only if $m_H L^2/n\to0$.

\section[Comparison with common initialization and tuning]{Comparison with common initialization and tuning}\label{sup:matched}
We generated 50 independent data sets for each combination of $H\in\{0,1\}$
and sign-reversal or mimic attacks, using the original simulation generator,
$m=20,n=90,p=30,s=5$, and $\alpha=0.2$. Within each data set all four methods
start from exactly the same attacked median-and-thresholded local pilot.
They use the same $h=0.24$, four rounds, trim fraction $0.25$, threshold floor
$0.042$, threshold schedule and step-norm cap. Local pilot penalties depend
only on the common site size and $p$, not on fault identities.

The only algorithmic change is the preconditioner: the local ridge inverse
$(\widehat A_k+0.08I)^{-1}$ is compared with $cI$ for the three prespecified
values $c=1,4,8$. All variants send parameter messages
$\beta-P_kg_k$ before trimming, so attacks operate in the same message units.
At each round, each method uses the same attack family and random seed;
the actual attack values can differ because the honest messages differ.
Thus this comparison measures the effect of the preconditioner, including
its interaction with attacks and the common cap. It does not isolate
heterogeneity correction from ordinary optimization preconditioning.

\begin{table}[ht]
\centering\small
\caption{Matched four-round comparison: mean $\ell_2$ error (Monte Carlo
standard error), 50 paired replications per setting.}
\begin{tabular}{llrrrr}
\toprule
$H$ & Attack & Local inverse & $I$ & $4I$ & $8I$\\
\midrule
0 & Sign & .151 (.002) & .296 (.003) & .300 (.006) & 1.307 (.009)\\
1 & Sign & .183 (.003) & .452 (.004) & .176 (.004) & 1.532 (.031)\\
0 & Mimic & .105 (.002) & .243 (.003) & .305 (.008) & 1.417 (.012)\\
1 & Mimic & .130 (.004) & .377 (.005) & .577 (.048) & 2.289 (.025)\\
\bottomrule
\end{tabular}
\end{table}

Local calibration has smaller mean error than the unit-step comparison in
all four settings, but it does not dominate every scalar step. For $H=1$
under sign reversal, the paired error difference ($4I$ minus local inverse)
is $-0.00685$ with Monte Carlo standard error $0.00358$. This matched
experiment supports an advantage over the unit step, but not uniformly over
all scalar steps. It is not an equal-wall-time comparison or a study of
converged optimizers.

The reproducibility archive includes \texttt{run\_matched.py}, every seed and round-level
error, honest-message dispersion, cap count, and paired summary. Reproduce from
its code directory with
\texttt{python run\_matched.py --reps 50 --workers 4}.
Seeds are $202609150+1000j+i$ for settings $j=0,1,2,3$ in the order shown
above and replications $i=0,\ldots,49$; attack seeds add $100000$ at initialization
and $200000+t$ at round $t=0,\ldots,3$. The three scalar steps are all reported;
no best-step selection is used in the summary.

\begingroup
\interlinepenalty=10000
\bibliography{references}
\endgroup

\end{document}